\documentclass[reqno,11pt,a4paper]{amsart}
\usepackage{csquotes}
\usepackage{amssymb}
\usepackage{mleftright}  %together with command \mleftright removes spurious space associated with \left, \right
\usepackage{xfrac}
\usepackage{mathrsfs}
\usepackage[dvipsnames]{xcolor}
\usepackage[colorlinks,linkcolor=RedViolet,citecolor=PineGreen,anchorcolor=Green,urlcolor=RedViolet,bookmarks=true]{hyperref}
\usepackage[headings]{fullpage}
\usepackage[numbers,square]{natbib}
\usepackage[shortlabels,inline]{enumitem}
\usepackage{anyfontsize}
\usepackage{scalerel}
\usepackage{bbm}
\usepackage{bm}
\usepackage{tikz}
\usepackage[percent]{overpic}  
\usepackage[foot]{amsaddr}
\usepackage[hang,flushmargin]{footmisc}

\makeatletter
\newcommand{\optionaldesc}[2]{%
  \phantomsection
  #1\protected@edef\@currentlabel{#1}\label{#2}%
}
\makeatother

\newcommand{\SR}{\textsf{\upshape SR}}
\newcommand{\MSR}{\textsf{\upshape MSR}}
\newcommand{\HR}{\textsf{\upshape HR}}
\newcommand{\MHR}{\textsf{\upshape MHR}}
\newcommand{\E}{\textsf{\upshape E}}
\newcommand{\Var}{\textsf{\upshape Var}}

\renewcommand{\P}{\textsf{\upshape P}}
\newcommand{\Q}{\textsf{\upshape Q}}
\newcommand{\indicator}[1]{\mathbf{1}_{#1}}

\newcommand{\Exp}{\mathscr{E}}
\newcommand{\Log}{\mathcal{L}}
\newcommand{\sint}{\bm{\cdot}}

\newcommand{\lc}{[\![}
\newcommand{\rc}{]\!]}
\newcommand*{\bigs}[1]{\scalebox{1.2}{\ensuremath#1}}
\newcommand{\vfrac}[2]{\ensuremath{#1 / #2}}

\newcommand{\N}{\mathbb{N}}
\newcommand{\R}{\mathbb{R}}

\newcommand{\id}{{\operatorname{id}}}

\newcommand{\I}{\mathfrak{I}}
\newcommand{\tI}{\widetilde{\mathfrak{I}}}
\newcommand{\MV}{\textnormal{\textsc{mv}}}
\newcommand{\MMV}{\textnormal{\textsc{mmv}}}
\newcommand{\aMV}{\hat{\alpha}^{\MV}}
\newcommand{\bMV}{\hat{\beta}^{\MV}}
\newcommand{\aMMV}{\hat{\alpha}}
\newcommand{\bMMV}{\hat{\beta}}
\newcommand{\xMMV}{\hat{x}}
\newcommand{\lMV}{\hat{\lambda}^{\MV}}
\newcommand{\lMVn}{\hat{\lambda}^{\MV,n}}
\newcommand{\lMMV}{\hat{\lambda}}
\newcommand{\lMMVn}{\hat{\lambda}^{n}}

\newcommand{\uMV}{u_{\MV}}
\newcommand{\vMV}{v_{\MV}}
\newcommand{\uMMV}{u_{\MMV}}
\newcommand{\vMMV}{v_{\MMV}}
\newcommand{\QMMV}{\hat{\Q}}
\newcommand{\QMV}{\hat{\Q}^{\MV}}
\newcommand{\AdmMMV}{\Theta_{\MMV}}
\newcommand{\AdmMV}{\Theta_{\MV}}

\renewcommand{\d}{\mathrm{d}}
\DeclareMathOperator*{\argmax}{arg\,max}

\newcommand{\util}{g}

\newcommand{\locutil}{\mathfrak{g}}
\newcommand{\locutilMMV}{\mathfrak{g}}
\newcommand{\locutilMV}{\mathfrak{g}^{\MV}}
\newcommand{\utilMMV}{\util_{\MMV}}
\newcommand{\utilMV}{\util_{\MV}}

\newcommand{\e}{\mathrm{e}}

\newcommand\shorturl[1]{%
  \href{http://#1}{\nolinkurl{#1}}%
}

\newtheorem{theorem}{Theorem}[section]
\theoremstyle{plain}

\newtheorem{assumption}[theorem]{Assumption}

\newtheorem{corollary}[theorem]{Corollary}

\newtheorem{definition}[theorem]{Definition}
\newtheorem{example}[theorem]{Example}

\newtheorem{lemma}[theorem]{Lemma}

\newtheorem{proposition}[theorem]{Proposition}
\newtheorem{remark}[theorem]{Remark}

\numberwithin{equation}{section}

\begin{document}
\title[{}]{Dynamically optimal portfolios for monotone mean--variance preferences}
\author{Ale\v{s} \v{C}ern\'{y}$^{a,b}$}
\address{$^a$Bayes Business School, City St George's, University of London\\
106 Bunhill Row, London EC1Y 8TZ, UK\vspace{-1.1ex}}
\email{\!ales.cerny@citystgeorges.ac.uk}
\author{Johannes Ruf$^b$}
\address{$^b$Department of Mathematics\\
  London School of Economics and Political Science\\
	Houghton Street, London, WC2A 2AE, UK\vspace{-1.1ex}}
\email{\!j.ruf@lse.ac.uk}
\author{Martin Schweizer$^c$}
\address{$^c$Department of Mathematics\\
 ETH Zurich\\
	R\"amistrasse 101, CH-8092 Zurich, Switzerland\vspace{0.4ex}}
\email{\!martin.schweizer@math.ethz.ch}
\date{March 6th, 2026}

\begin{abstract}
Monotone mean--variance (MMV) utility is the minimal modification of the classical Markowitz mean--variance (MV) utility that respects rational ordering of investment opportunities. This paper provides, for the first time, a complete characterization of optimal dynamic portfolio choice for the MMV utility in asset price models with independent returns. The task is performed under minimal assumptions, weaker than the existence of an equivalent martingale measure and with no restrictions on the moments of asset returns. We interpret the maximal MMV utility in terms of the monotone Sharpe ratio (MSR) and show that the global squared MSR arises as the nominal yield from continuously compounding at the rate equal to the maximal local squared MSR. The paper gives simple necessary and sufficient conditions for MV efficient portfolios to be MMV efficient. Several illustrative examples contrasting the MV and MMV criteria are provided.\\[1ex]
\end{abstract}
\vspace*{-1.2cm}
\maketitle
\mleftright
\section{Introduction}

The concept of mean--variance efficiency does not always prefer more to less. For example, a wealth distribution $W_A$ with values $-0.3$, $0.6$, $1.2$ with equal probabilities has higher ratio of mean to standard deviation (so-called Sharpe ratio \cite{sharpe.94}) than $W_B$ with values $-0.3$, $0.6$, $2.4$. Thus, a hypothetical market where one can trade multiples of $W_A$ at zero price has a better mean--variance trade-off than the parallel market trading only multiples of $W_B$ at price zero. However, a rational investor would nonetheless prefer market $B$ to market $A$ as they could take $W_B$, set aside or give away the amount $0$, $0$, $0.9$, and obtain a Sharpe ratio that exceeds that of $W_A$.

In view of this example, we shall say that a trading position $W$ in a market $M$ is \emph{monotone mean--variance (MMV) efficient} if there is $Y\geq 0$ such that $W-Y$ has the lowest variance for a given mean. That is, for all $W'\in M$ and $ Y'\geq0$ such that $\E[W-Y]=\E[W'-Y']$, one has $\Var(W-Y)\leq\Var(W'-Y')$, where the comparison is performed over $Y$ and $Y'$ such that $W-Y\in L^2$ and $W'-Y'\in L^2$.
This notion of efficiency is implicit in the work of \citet*{cui.al.12}, for example.

We shall now discuss preferences that are consistent with monotone mean--variance efficiency. Consider the classical Markowitz mean--variance utility $V_{\textsc{mv}}: L^2 \to \R$ given by 
\begin{align}\label{eq:MV1}
V_{\textsc{mv}}(W) = \E[W]-\frac{1}{2}\Var(W) 
\end{align}
and its minimal monotone modification given by
\begin{equation}\label{eq:MMV1}
V_{\MMV}(W) = \sup_{Y\geq0} V_{\MV}(W-Y),
\end{equation}
where the supremum is taken over $Y\geq0$ such that $W-Y\in L^2$. The \emph{monotone mean--variance utility} $V_{\MMV}$ sets aside a non-negative amount $Y$ to make the mean--variance utility of the remaining wealth $W-Y$ as high as possible. Observe that thanks to \eqref{eq:MMV1} the domain of the MMV utility extends beyond $L^2$, which will later allow us to consider returns that are not square integrable.

Our aim is to compute dynamically optimal portfolios for the monotone mean--variance utility in models with independent, but not necessarily time-homogeneous, returns. The purpose of doing so, apart from obtaining the convenience of explicit formulae, is to gain better understanding of how the global optimal portfolio performance arises from the local conditions in the market. The literature in this area is very sparse and we shall return to it in due course in Section~\ref{S:5}. One of the major obstacles in obtaining clean results is that neither $V_{\MV}$ nor $V_{\MMV}$ are easily amenable to dynamic programming. To address this issue, observe that the classical mean--variance utility \eqref{eq:MV1} is the cash-invariant hull of the expected quadratic utility, i.e., 
\begin{equation}\label{eq:MV2}
V_{\MV}(W) ={}  \sup_{ \eta\in \R }\, \left\{\E\left[\util_{\MV}(W-\eta)\right]+\eta\right\},\quad W \in L^2,
\end{equation}
with the quadratic utility given by
\begin{equation}\label{eq:utilMV}
\util_{\MV} = \id - \frac{1}{2} \id^2.
\end{equation}
Here $\id$ is an identity function, $\id(x)=x$. Since the operations of taking the monotone hull in \eqref{eq:MMV1} and the cash-invariant hull in \eqref{eq:MV2} commute, it is shown in \citet*[Equation (9)]{cerny.20} that $V_{\MMV}$ is the cash-invariant hull of the expected monotone quadratic utility, i.e., 
\begin{equation}\label{eq:MMV2}
V_{\MMV}(W) ={}  \sup_{ \eta\in \R }\, \left\{\E\left[\util_{\MMV}(W-\eta)\right]+\eta\right\},\quad W \in L^0_+-L^2_+,
\end{equation}
with the monotone quadratic utility given by
\begin{equation}\label{eq:utilMMV}
\util_{\MMV} =\id\wedge1 - \frac{1}{2} (\id\wedge1)^2;
\end{equation}
see also Figure~\ref{F:1}. 
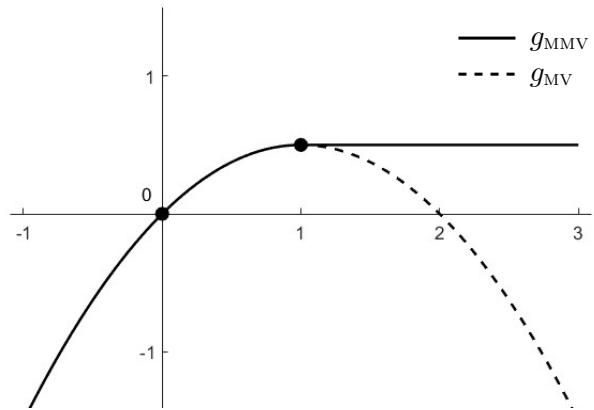
\begin{figure}
	\centering
	\begin{overpic}[trim=20 90 2 30,clip,scale=.4]{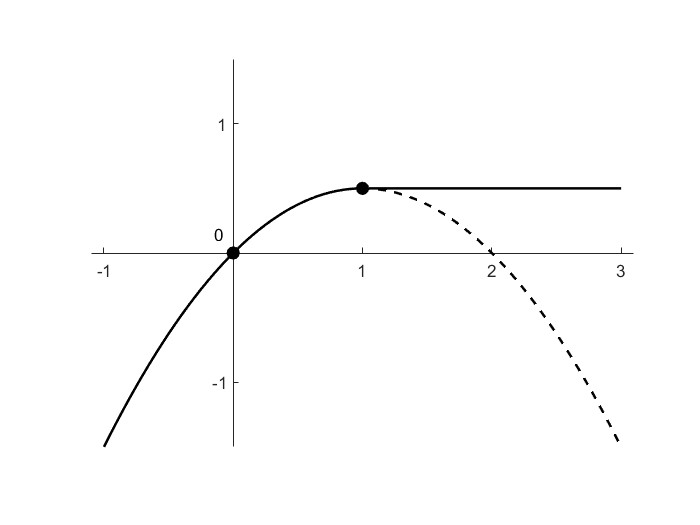}
\put(72,51){\tikz \draw[line width=0.4mm,black] (0,0)--(0.75,0);}
\put (82,50.3) {\textcolor{black}{$\displaystyle \util_{\MMV}$}}
\put(72,46){\tikz \draw[line width=0.4mm,black,dashed] (0,0)--(0.75,0);}
\put (82,45.3) {\textcolor{black}{$\displaystyle \util_{\MV}$}}
\end{overpic}
	\caption{Quadratic utility 	and its monotonization.} 
	\label{F:1}
\end{figure}

In the handful of existing studies concerned with explicit optimal MMV utility portfolio selection in continuous time (e.g., \cite{hu.shi.xu.23.sifin}), the key formula \eqref{eq:MMV2} is hidden behind specific model assumptions. Formula \eqref{eq:MMV2} shows in complete generality that the MMV utility is not very far from expected monotone quadratic utility. This offers a direct line of attack, which we shall exploit fully below. 

The rest of the paper proceeds as follows.  In Section~\ref{S:2}, we set up our framework (Subsection~\ref{SS:2.1}) and thereafter show that 
\begin{enumerate}[(a)]
\item The MMV utility maximization can be recast in terms of the time-consistent expected monotone quadratic utility maximization (Subsection~\ref{SS:2.2}). This step is model independent and only requires that the admissible trading strategies form a cone.
\item Trading strategies have a useful parametrization in terms of dollar investment per unit of risk tolerance (Subsection~\ref{SS:2.3}). This step, too, is model independent.
\item The global expected utility maximization can be recast in terms of simpler local expected utility. The local expected utility admits a maximizer under conditions weaker than the existence of an equivalent $\sigma$-martingale measure (Subsection~\ref{SS:2.4}). From here onwards, instantaneous absence of arbitrage and independent increments are assumed.
\item The local maximizer integrates the yield process if and only if its cumulative local expected utility is finite (Subsection~\ref{SS:2.5}). This step is trivial in time-homogeneous (i.e., L\'evy) models but has not been successfully tackled for general independent returns for any local expected utility to date, to the best of our knowledge.
\item If its cumulative local expected utility is finite, the local optimizer also generates the optimal trading strategy with finite MMV utility. If not, there is an $L^2$ free lunch in the form of a sequence of near-arbitrage opportunities whose loss goes to $0$ in $L^2$ and whose gain is above $1$ with probability approaching $1$; at the same time, the MMV utility is unbounded (Subsection~\ref{SS:2.6}).
\item The MMV utility of dynamic trading is finite if and only if the yield process admits a separating measure whose density has the smallest variance. Furthermore, simple necessary and sufficient conditions for this separating measure to be a $\sigma$-martingale measure are given (Subsection~\ref{SS:2.7}).
\end{enumerate}
Section~\ref{S:3} discusses the economic interpretation of the main results in terms of monotone Sharpe ratios. In particular, we show here for the first time in the literature how the global MMV portfolio performance is generated by local monotone Sharpe ratios. Section~\ref{S:4} provides several illustrative examples, Section~\ref{S:5} relates the results of this paper to previous studies, and Section~\ref{S:6} concludes. 

The results in this paper rely on the ``$\circ$'' operation, whereby a (predictable) function $g$ with $g(0)=0$ acts on the increments of some semimartingale $X$ to form a new process $g\circ X$ that one may loosely think of as the $g$-variation of $X$. For example, $[X,X]=\id^2\circ X$ denotes the quadratic variation, with $\id$ standing for an identity function. The details of the $\circ$ calculus are worked out in \citet{cerny.ruf.22.ejp} with further discussion and applications available in \cite{cerny.ruf.21.ejor,cerny.ruf.23.spa}. Very briefly, the $\circ$ operation provides a convenient way of manipulating stochastic processes, e.g., the ``Yor formula'' takes the form $\Exp(g\circ X)^p=\Exp(((1+g)^p-1)\circ X)$ for all $p\geq 2$, where $\Exp$ denotes the \citet{doleans-dade.70} stochastic exponential.  We refer the reader to the introduction of \cite{cerny.ruf.22.ejp} and \cite[Section~2]{cerny.ruf.23.spa} for executive summaries of the $\circ$ calculus.

The calculus also provides a systematic way of computing the drift rate $b^{g\circ X}$ of the $g$-variation in terms of the ``characteristic triplet'' of the underlying process $X$, written here as $(b^{X[1]},c^X,F^X)$ relative to some activity process $A$. Intuitively, $b^{X[1]}$ is the drift rate of the process $X$ with large jumps removed; $c^X$ is the squared volatility; and $F^X$ is a measure describing the jump intensity of different jump sizes. We refer the reader to  \citet[II.2.6; II.2.9]{js.03} for the definition of semimartingale characteristics, to \citet{kallsen.04} for a pedagogical introduction, and to Examples~3.1, 3.4, 4.1, and 4.3 in \cite{cerny.ruf.21.ejor} for illustrations.

%=================%
%=================%

\section{Dynamically optimal portfolio selection}\label{S:2}

%=================%
\subsection{Asset prices and admissibility}\label{SS:2.1}
We assume there is a risk-free asset with constant value 1 and $d\in\N$ risky assets with cumulative dollar yield $R$ starting at $0$. For example, if $S$ denotes the value of the $d$ risky assets such that $S_->0$, the yield is given componentwise by  
\[
R^{(i)} = \int_0^\cdot \frac{\d S^{(i)}_t}{S^{(i)}_{t-}},\qquad i\in\{1,\ldots,d\}.
\] 
Whereas \citet*{eberlein.jacod.97,kallsen.00,jakubenas.02,esche.schweizer.05,jeanblanc.al.07,bender.niethammer.08} and \citet*{nutz.12.mafi} take the cumulative yield to be a L\'evy process satisfying $\Delta R>-1$, we shall consider $R$ from the wider class of semimartingales with independent increments and impose no lower bound on $\Delta R$. The assumption of independent increments will be invoked from Subsection~\ref{SS:2.4} onwards.

We sometimes write $L(R)$ for the set of predictable processes that integrate $R$. For $\vartheta\in L(R)$, we often use the shortcut notation $\vartheta\sint R = \int_0^\cdot \vartheta_t\d R_t$. Here $\vartheta$ is interpreted as a row vector of dollar investments in the risky assets and the integral $\vartheta\sint R$ as the net gain of the corresponding self-financing strategy.   For a  function $g:\R\to\R$ twice continuously differentiable at zero and $\vartheta\in L(R)$, we call 
$$g\circ (\vartheta\sint R) = g'(0)\vartheta\sint R +\frac{1}{2}g''(0)[\vartheta\sint R,\vartheta\sint R]^c + \sum_{t\leq\cdot}\bigs(g(\vartheta_t\Delta R_t)-g'(0)\vartheta_t\Delta R_t\bigs) $$
the $g$-variation of $\vartheta\sint R$. Occasionally, we use the operation ``$\circ$'' with predictable functions, relying on the results in \citet*{cerny.ruf.22.ejp}. We write $\id$ for the identity function, i.e., $\id(x)=x$ and $\id_1$, $\id_2$, etc.\ for its components, i.e., $\id_1(x)=x_1$, $\id_2(x)=x_2$, etc\@. We also fix a truncation function $h:\R^d\to\R^d$ given by 
\begin{equation}\label{eq:truncation}
h = (\id_1\indicator{|\id_1|\leq 1},\ldots,\id_d\indicator{|\id_d|\leq 1}),
\end{equation}
and write  $R[1]=h\circ R$ for the process $R$ with large jumps removed.

For a semimartingale $R$ with independent increments, by \citet*[Proposition~II.2.9 and Theorem~II.4.15]{js.03}
there is a deterministic, right-continuous, non-decreasing process $A$ with left limits such that the characteristics of $R$ are absolutely continuous with respect to $A$ and the resulting differential characteristics $(b^{R[1]},c^R,F^R)$ are deterministic. 
For a specific example of differential characteristics in a bivariate Merton jump-diffusion model, see, e.g., Example~4.3 in \cite{cerny.ruf.21.ejor}. In all It\^o semimartingale models one may choose the activity process $A_t=t$.
Specifying $b^{R[1]}$ rather than $b^R$ allows for the possibility that $R$ does not have a finite mean.

From Subsection~\ref{SS:2.4} onwards, we shall make the following standing assumption, which is weaker than the existence of an equivalent $\sigma$-martingale measure for $R$.

\begin{assumption}\label{A:instNA_PII}
The yield process $R$  
is \emph{instantaneously arbitrage-free}, i.e., for all $t\in[0,T]$ and all  
$\lambda\in\R^d$ one has that 
\begin{equation}\label{eq:instNA_PII}
\left.
\begin{gathered}
 \left(\lambda c^R_t = 0 \text{ and } \int_{\R^d} \indicator{\lambda x< 0}F^R_t(\d x) = 0\right) \quad \implies \\
\left(\int_{\R^d} \indicator{0<\lambda x}F^R_t(\d x) > 0 \text{ and }
\lambda b^{R[1]}_t-\int_{\R^d}  \lambda h(x) F^R_t(\d x)<0\right)\\
\quad\text{or}\quad \\
\left(\int_{\R^d} \indicator{0<\lambda x}F^R_t(\d x) = 0 \text{ and } \lambda b^{R[1]}_t=0\right).
\end{gathered}
\qquad\right\}
\end{equation}
\end{assumption}
Intuitively, Assumption~\ref{A:instNA_PII} requires that process $\lambda R$ be free to move in both directions (up and down) or else be constant at any moment in time for any $\lambda\in\R^d$. At times $t$ when $R$ jumps with non-zero probability, the above condition is tantamount to the classical one-period no-arbitrage requirement
\begin{equation}\label{eq:NA_t}
\P[\lambda \Delta R_t \geq 0]=1 \Rightarrow \P[\lambda \Delta R_t=0]=1,\qquad \lambda\in\R^d,\ t\in[0,T]. 
\end{equation}
When $\P[\Delta R_t\neq 0]=0$, the definition requires that for each $\lambda\in\R^d$, the differential characteristic triplet of $\lambda R$ taken at $t$ is a triplet of a L\'evy process that is not a subordinator; see \citet*[Lemma~3.5]{kardaras.09}.   

\begin{definition} We say that a trading strategy $\vartheta\in L(R)$ is admissible, writing $\vartheta\in\AdmMMV$, if $\sup_{t\in[0,T]}(\vartheta\sint R_t)^-\in L^2$. We say that $\vartheta\in L(R)$ is mean--variance admissible, if the stricter requirement 
$\sup_{t\in[0,T]}|\vartheta\sint R_t|\in L^2$ holds. For such strategies we write $\vartheta\in\AdmMV$. 
\end{definition}
Observe that $\AdmMMV$ is a cone while $\AdmMV$ is a linear subspace inside $\AdmMMV$. Furthermore, $\AdmMMV$ contains, among others, all strategies whose wealth is bounded below by a constant. For later use we remark that for $\vartheta \in L(R)$ one can write more symmetrically
\begin{align}
\vartheta\in \AdmMMV \iff {}&{}\sup_{t\in[0,T]}|\utilMMV'(\vartheta\sint R_t)|\in L^2,\label{eq:250111.MMV}\\
\vartheta\in \AdmMV \iff {}&{} \sup_{t\in[0,T]}|\utilMV'(\vartheta\sint R_t)|\in L^2.\label{eq:250111.MV}
\end{align}

%=================%
\subsection{First steps}\label{SS:2.2}

For the monotone quadratic utility function $\utilMMV$ in \eqref{eq:utilMMV}, consider the optimization problem
\begin{equation}\label{eq:uMMV}
\uMMV(x)=\sup_{\vartheta \in \AdmMMV}\{\E[\utilMMV(x+\vartheta \sint R_T)]\},\qquad x\in\R.
\end{equation}
Denote an optimal trading strategy in \eqref{eq:uMMV}, should it exist, by $\aMMV(x)$. Observe that maximization of expected utility is a time-consistent problem.

Recall next the monotone mean--variance utility in \eqref{eq:MMV1} and consider the corresponding optimal portfolio allocation problem 
\begin{equation} \label{eq:vMMV}
\vMMV(x) = \sup_{\vartheta \in \AdmMMV}\, \{V_{\MMV}(x+\vartheta \sint R_T)\},\qquad x\in\R. 
\end{equation}
The optimal strategy, should it exist, will be denoted $\bMMV(x)$, $x\in\R$. 

We will now recall results from \cite[Theorem~5.1]{cerny.20} that (a) allow $\vMMV(0)$ to be expressed directly in terms of $\uMMV(0)$; (b) identify an optimal trading strategy in \eqref{eq:vMMV} by means of the time-consistent zero-cost optimal strategy $\aMMV(0)$ from \eqref{eq:uMMV}. The proof is found in Appendix~\ref{appx:A}. The next proposition is model-free in the sense that it works for any semimartingale $R$ (i.e., also without independent increments) and any set of strategies that forms a cone. Furthermore, no assumptions on the absence of arbitrage are needed.  
\begin{proposition}\label{P1}
The following are equivalent.
\begin{enumerate}[(i)]
\item\label{T1.i} $2\uMMV(0)<1$;
\item\label{T1.ii} $\vMMV(0)<\infty$.
\end{enumerate}
Furthermore, if either of the two conditions holds, then 
\begin{equation}\label{eq:vMMVx}
\vMMV(x) = x+\frac{\left(1-2\uMMV(0)\right)^{-1}-1}{2},\qquad x\in\R,
\end{equation}
and \eqref{eq:vMMV} admits an optimizer for $x=0$ if and only if \eqref{eq:uMMV} does so for $x=0$, in which case \eqref{eq:vMMV} admits an optimizer for each $x\in\R$ and
\begin{equation}\label{eq:bMMV}
\bMMV(x) = \left(1+2\vMMV(0)\right)\aMMV(0),\qquad x\in\R.
\end{equation}
\end{proposition}
\begin{remark}\label{R:MSR}
By \citet*[Eq. (17)]{cerny.20}, the quantity $2\vMMV(0)$ is the square of the maximal monotone Sharpe ratio available by trading in $\AdmMMV$ with zero initial cost. 
If the optimal strategy exists, the Sharpe  ratio of the truncated wealth $(\aMMV(0)\sint R_T) \wedge 1$ attains the maximal value $\sqrt{2\vMMV(0)}$ and $(\aMMV(0)\sint R_T-1)\vee0$ is the free cash flow stream in the terminology of \citet{cui.al.12}. For more discussion in this direction, see Section~\ref{S:3}. 
\end{remark}

\begin{remark}\label{R:risk aversion 1}
Observe that the absolute risk aversion of $\utilMV$ at $0$ is $1$. For $\gamma>0$, one could instead consider $\widetilde\utilMV = \utilMV(\gamma\id)/\gamma$ with the corresponding monotone hull $\widetilde\utilMMV = \utilMMV(\gamma\id)/\gamma$, which now both have absolute risk aversion of $\gamma$ at $0$. On defining a new MMV utility via \eqref{eq:MMV2} with $\utilMMV$ replaced by $\widetilde\utilMMV$ and letting $\widetilde{\beta}(x)$ be the optimal MMV strategy, it is immediate that 
$\widetilde{V_{\MMV}} = V_{\MMV}(\gamma\id)/\gamma$, 
hence $\widetilde{\beta}(x) = \hat\beta(0)/\gamma$ and  $\widetilde{v_{\MMV}}(x)-x = \vMMV(0)/\gamma$ for all $x\in\R$, i.e., the optimal risky investment  and the certainty equivalent gain are inversely proportional to the local absolute risk aversion of the MMV utility at $0$.

This shows that the MMV efficient frontier is spanned by the risk-free investment plus an arbitrary positive multiple of the optimal zero-cost strategy $\aMMV(0)$. After disposing of high wealth optimally as shown in Remark~\ref{R:MSR}, the slope of the efficient frontier on the mean--standard deviation diagram is given by the monotone Sharpe ratio $\sqrt{2\vMMV(0)}$.

Equation \eqref{eq:bMMV} further shows that the optimal risky MMV portfolio for MMV risk aversion $\gamma$ is the same as the optimal portfolio for expected monotone quadratic utility with initial wealth zero and risk aversion $\gamma/(1+2\vMMV(0))$ at $0$.  
\end{remark}

%=================% 
\subsection{Parametrization of trading strategies}\label{SS:2.3}
The economic theory of risk aversion suggests a useful parametrization of monotone mean--variance trading rules. Consider some trading strategy $\alpha\in L(R)$ with wealth process $W=\alpha\sint R$. Observe that in the expected utility problem \eqref{eq:uMMV}, it is optimal \emph{not} to invest in the risky assets whenever the associated wealth exceeds the bliss point $1$ of the monotone quadratic utility function $\utilMMV$. To underscore this point, the local absolute risk tolerance, 
\[
-\frac{\utilMMV'(W_{t-})}{\utilMMV''(W_{t-})}\indicator{\{W_{t-}<1\}}=\left(1 - W_{t-}\right)^+=\utilMMV'(W_{t-}),
\] 
falls towards zero as the wealth $W_{t-}$ approaches the bliss point $1$ from below. It therefore makes sense to parametrize the trading strategy by $\lambda_t$, the dollar amount in the risky assets per unit of local risk tolerance. This then yields wealth dynamics of the form
\begin{equation}\label{eq:210401.1}
	\d W_t = \lambda_t \left(1 - W_{t-}\right)^+ \d R_t
\end{equation}
for the optimal net wealth process, for a suitably chosen predictable $\lambda\in L(R)$. For a fixed $\lambda\in\R^d$, this very much resembles a \emph{constant proportion portfolio insurance} (CPPI) strategy of \citet*{black.perold.92}, except instead of a floor on losses there is a ceiling on gains; the riskiness of the investment is reduced in proportion to the gap (known as \emph{cushion} in the CPPI literature) between the net gain $W$ and the ceiling $1$.

By Lemma~\ref{L:SDE}, the unique solution of the stochastic differential equation \eqref{eq:210401.1} has the property
\begin{equation}\label{eq:210401.2}
\left(1 - W\right)^+=\Exp\left( \left(-1\vee \id\right) \circ (-\lambda\sint R) \right) = \Exp\left( -(\id\wedge1) \circ (\lambda\sint R) \right).
\end{equation}
Here, $\Exp(X)$ denotes the stochastic exponential of a semimartingale $X$. Economically, $\Exp(X)$ describes the value of an investment with initial value $1$, continuously compounded using the stochastic cumulative yield $X$. Intuitively, the process $(\id\wedge1) \circ (\lambda\sint R)$ truncates the jumps of the yield $\lambda\sint R$ at $1$. Combining \eqref{eq:210401.1} and \eqref{eq:210401.2}, the net wealth has the form
$W = \alpha \sint R$, where to each $\lambda\in L(R)$ we have associated the trading strategy (dollar investment in stock)
$
 \alpha = \lambda\Exp\left( -\left(\id\wedge1\right) \circ (\lambda\sint R) \right)_{-}
$.

The next proposition explores these ideas and their converse in detail. The proof is found in Appendix~\ref{appx:B}. This proposition, too, works for $R$ without independent increments. We say that a semimartingale $X$ ``does not go to zero continuously'' if for the stopping time $\sigma= \inf\{t > 0 : X_t = 0\}\wedge T$ one has $X_-\neq 0$ on $\lc 0,\sigma\rc$.

\begin{proposition}\label{P:lambda to alpha and back}
Suppose $\lambda$ integrates $R$. Then
\begin{equation}\label{eq:alpha}
\alpha = \lambda\Exp\left( -\left(\id\wedge1\right) \circ (\lambda\sint R) \right)_{-}
\end{equation}
integrates $R$ and for the (marginal) utility of the wealth $\alpha\sint R$ one has 
\begin{align}
\utilMMV'(\alpha\sint R) ={}& (1-\alpha\sint R)^+=\Exp\left( -\left(\id\wedge1\right) \circ (\lambda\sint R) \right),\label{eq:Granada241014}\\
2\utilMMV(\alpha\sint R) ={}& 1-\Exp(-2\utilMMV\circ (\lambda\sint R)).\label{eq:terminal utility}
\end{align}
Furthermore,
\begin{enumerate}[(a)]
\item\label{20241230.a} the marginal utility $\utilMMV'(\alpha\sint R)=(1-\alpha\sint R)^+$ does not go to zero continuously;
\item\label{20241230.b} $\alpha$ is zero from the first time the marginal utility $\utilMMV'(\alpha\sint R)$ becomes zero.
\end{enumerate}
Conversely, assume that $\alpha$ integrates $R$ and that \ref{20241230.a} and \ref{20241230.b} hold. 
Denote by $\sigma$ the first time $(1-\alpha\sint R)^+$ becomes 0. Then 
\begin{equation}\label{eq:Cierne250101}
 \lambda = \frac{\alpha}{1-\alpha\sint R_-}\indicator{\lc0,\sigma\rc}
\end{equation}
is well defined, it integrates $R$, and \eqref{eq:alpha} holds.
\end{proposition}

%=================% 
\subsection{Optimal dollar investment per unit of local risk tolerance}\label{SS:2.4}
From now on to the end of the paper, we assume that $R$ has independent increments and we invoke the instantaneous no-arbitrage condition in Assumption~\ref{A:instNA_PII}. All proofs of this subsection are found in Appendix~\ref{appx:C}.

For deterministic $\lambda$, we shall now produce an explicit formula for the expected monotone quadratic utility of the strategy $\alpha$ in \eqref{eq:alpha}. Recall that a semimartingale is called `special' if it is a sum of a finite variation predictable process and a local martingale; see \cite[I.4.21--23]{js.03}.
\begin{proposition}\label{P:Cierne240807}
For a deterministic $\lambda$ that integrates $R$, the following are equivalent.
\begin{enumerate}[(i)]
\item\label{P:Cierne240807.i} The strategy $\alpha$ in \eqref{eq:alpha} has finite expected utility, i.e., $\utilMMV(\alpha\sint R_T)\in L^1$.
\item\label{P:Cierne240807.ii} The realized local utility $\utilMMV \circ (\lambda\sint R)$ is special.
\end{enumerate}
Furthermore, if one of these conditions holds, then
\begin{equation}\label{eq:Cierne240808}
\E[2\utilMMV(\alpha\sint R_T)] = 1-\Exp\bigs( b^{-2\utilMMV \circ\mkern1mu (\lambda\sint R)} \sint A\bigs)_T<1.
\end{equation}
\end{proposition}

\begin{remark} The quantity $\utilMMV\circ (\lambda\sint R)$ appearing on the right-hand side of \eqref{eq:Cierne240808} is the $\utilMMV$-variation of the cumulative yield $\lambda\sint R$. Since the variation is a limit of partial sums with summands of the form $\utilMMV(\lambda\sint R_{t_{n+1}}-\lambda\sint R_{t_n})$ (e.g., \citet*[Th\'eor\`eme~2a]{emery.78}), in Proposition~\ref{P:Cierne240807} we have interpreted this quantity as the \emph{realized local utility} of the strategy $\lambda$.
\end{remark}

A conditional expectation analog of \eqref{eq:Cierne240808} shows that 
\[
\E_t[2\utilMMV(\alpha\sint R_T)] = 
1-(\utilMMV'(\alpha\sint R_t))^2 \Exp\bigs( \indicator{\rc t,T\rc}b^{-2\utilMMV \circ\mkern1mu (\lambda\sint R)} \sint A\bigs)_T,
\]
i.e., the conditional expected utility generated by deterministic processes $\lambda$ is a deterministic affine transformation of $\utilMMV(W_t)$, whose local risk tolerance with respect to $W_t$ is therefore yet again $\left(1 - W_t\right)^+$. This strongly indicates that the optimal strategy, should it exist, is generated by a deterministic $\lambda$ and that  
\begin{quote}
the optimal $\lambda$ is obtained by maximizing the predictable $\utilMMV$-variation (i.e., the local expected utility) of the cumulative yield $\lambda\sint R$.
\end{quote}

We shall now examine the calculation of the optimal $\lambda$ in more detail. 
Consider a deterministic $\lambda$ that integrates $R$ and such that the realized local utility $\utilMMV \circ\mkern1mu (\lambda\sint R)$ is special. By composition (Proposition~\ref{P:composition}), one has 
$\utilMMV \circ\mkern1mu (\lambda\sint R) = \utilMMV(\lambda\mkern1mu\id) \circ R$. 
From Lemma~\ref{L:240130} one  obtains an explicit expression for the local expected utility rate 
\begin{equation*}
b^{\utilMMV \circ\mkern1mu (\lambda\sint R)}=  b^{ R[1]}\lambda - \frac{1}{2}\lambda c^{ R}\lambda^\top
+\int \bigs(\utilMMV(\lambda x)-\lambda h(x)\bigs)F^{ R}(\d x),
\end{equation*}
which then enters the expression \eqref{eq:Cierne240808}. 

\begin{definition}
We say that the deterministic function $\locutilMMV:[0,T]\times\R^d\to [-\infty,\infty)$ given by 
\begin{equation}\label{eq:locutil}
 \locutilMMV_t(\lambda) = b_t^{ R[1]}\lambda - \frac{1}{2}\lambda c_t^{ R}\lambda^\top
+\int_{\R^d} \bigs(\utilMMV(\lambda x)-\lambda h(x)\bigs)F_t^{ R}(\d x)
\end{equation}
is the \emph{local expected utility}.
\end{definition} 
Under the standing assumptions of independent returns and instantaneous absence of arbitrage, the local expected utility admits a maximizer over $\lambda$ $t$-by-$t$ for all $t$. Observe that $\locutilMMV$ is well defined even if $\lambda$ does not integrate $R$ or if $\utilMMV \circ\mkern1mu (\lambda\sint R)$ fails to be special. 

\begin{proposition}\label{P:Cierne240808}
There is a deterministic predictable process $\lMMV$ such that everywhere on $[0,T]$ one has  
\begin{enumerate}[(1)]
\item\label{Cierne240808.1} $\locutilMMV(\lambda)\leq \locutilMMV(\lMMV)$  for all $\lambda\in\R^d$;
\item\label{Cierne240808.2} $ 0\leq \locutilMMV(\lMMV)<\infty$; 
\item\label{Cierne240808.3} $\int_{\R^d} (\lMMV x)^2\indicator{\lMMV x<0}F^R(\d x) <\infty$ and 
\begin{equation*}
 b^{ R[1]}\lMMV - \lMMV c^{ R}\lMMV^\top
+\int_{\R^d} \bigs(\utilMMV'(\lMMV x)\lMMV x-\lMMV h(x)\bigs)F^{ R}(\d x) = 0.
\end{equation*}
\end{enumerate}
\end{proposition}
\begin{remark}
Proposition~\ref{P:Cierne240808} is novel to the extent that it establishes the existence of a maximizer of a specific local expected utility for a yield process with general independent increments. A precursor to this type of result can be found in \citet*{nutz.12.mafi} for L\'evy processes and \citet*{cerny.al.12} for one-period models. An expression of the form \eqref{eq:locutil} first appeared in \citet*[Definition~3.2]{kallsen.99} under the name ``local utility.''
\end{remark}
\begin{remark}\label{R:nonuniqueness}
If the model $R$ admits null strategies, i.e., predictable $\lambda\neq0$ such that $\lambda b^{R[1]}=0$, $\lambda c^R=0$, and $\int_{\R^d}\indicator{\lambda x\neq 0}F^R(\d x)=0$, then trivially $\lMMV$ can be picked in many different ways. But even in the absence of such strategies, $\lMMV$ may not be unique since $\utilMMV$ is not strictly concave. This can be seen already in a one-period model with four atoms (Example~\ref{E:1}). We fix one $\lMMV$ from now on and refer to it as \emph{the locally optimal strategy}. Observe that properties \ref{Cierne240808.1}--\ref{Cierne240808.3} of Proposition~\ref{P:Cierne240808} hold for any local optimizer, regardless of whether it is measurable.
\end{remark}
%=================% 
\subsection{Optimal local expected utility strategy integrates \texorpdfstring{$R$}{R}}\label{SS:2.5}
The aim of this section is to show that if the cumulative local expected utility is finite, then the optimal strategy $\lMMV$ integrates $R$. To the best of our knowledge, this is the first time that a link of this kind has been established in the literature.

One of the key ideas of the proof is to describe the maximal local expected utility $\locutil(\lMMV)$ as the drift rate of the variation $\utilMMV\circ (\lMMV\sint R)$. However, for the variation to be well defined, we first need to know that $\lMMV$ integrates $R$, leaving us with a conundrum how to break out of the circular reasoning. The second difficulty is that, even when granted that $\lMMV$ integrates $R$, we do not know whether $\utilMMV\circ (\lMMV\sint R)$ has a drift rate in the classical sense of being special. 

Our way around the first obstacle is to consider a sequence of strategies $\lMMVn = \lMMV\indicator{|\lMMV|\leq n}$, $n\in\N$. Here 
$|\mkern1mu\cdot\mkern1mu|$ stands for an arbitrary norm on $\R^d$. We overcome the second difficulty by assigning a drift rate to processes that are not necessarily special but only $\sigma$-special as described in Appendix~\ref{appx:D}. 
\begin{lemma}\label{L:armistice1}
The process $\utilMMV\circ (\lMMVn\sint R)$ is $\sigma$-special with 
\begin{equation*}
0\leq b^{\utilMMV\circ (\lMMVn\sint R)} =\locutilMMV(\lMMVn)\uparrow \locutilMMV(\lMMV),\qquad n\uparrow \infty.
\end{equation*}
\end{lemma}
\begin{proof}
Proposition~\ref{P:Cierne240808} yields $0\leq \locutilMMV(\lMMVn)\uparrow \locutilMMV(\lMMV)$. Lemma~\ref{L:240130} yields that $\utilMMV\circ (\lMMVn\sint R)$ is $\sigma$-special with $b^{\utilMMV\circ (\lMMVn\sint R)} =\locutilMMV(\lMMVn)$.
\end{proof}
Next, we shall recast integrability of $\lMMV$ in terms of predictable variations. The ability to do so is based on a new result developed for this purpose in Corollary~\ref{C:integrability}.
\begin{lemma}\label{L:armistice2}
The following are equivalent.
\begin{enumerate}[(i)]
\item The optimal local expected utility strategy $\lMMV$ integrates $R$.
\item The predictable process 
\begin{equation}\label{eq:London241109}
\lim_{n\uparrow \infty} \left(b^{(\id^2\wedge1)\circ (\lMMVn\sint R)} + \big|b^{\id\mkern1mu\indicator{|\id|\leq1}\circ (\lMMVn\sint R)} \big| \right)
\end{equation}
integrates $A$.
\end{enumerate}
\end{lemma}

With Lemmata~\ref{L:armistice1} and \ref{L:armistice2} in place, we continue by observing  that $\utilMMV$ in \eqref{eq:utilMMV}
is essentially linear--quadratic. Since the optimality of $\lMMV$ ties together the first and second moment of the truncated optimal yield, one obtains with a little bit of effort that both quantities in \eqref{eq:London241109} are controlled by the local expected utility $\locutilMMV(\lMMVn)$. The proofs of the next proposition and theorem are  found in Appendix~\ref{appx:E}.
\begin{proposition}\label{P:London241109}
For the optimal local expected utility strategy $\lMMV$ we have $b^{\id\mkern1mu\indicator{|\id|\leq1}\circ (\lMMVn\sint R)}\geq0$ and   
\begin{equation}\label{eq:key inequality}
2 b^{\utilMMV\circ (\lMMVn\sint R)}\leq b^{(\id^2\wedge1)\circ (\lMMVn\sint R)} + b^{\id\mkern1mu\indicator{|\id|\leq1}\circ (\lMMVn\sint R)} 
\leq 6 b^{\utilMMV\circ (\lMMVn\sint R)},\qquad n\in\N.
\end{equation}
\end{proposition}

The desired result is now immediate.
\begin{theorem}\label{T1}
The following are equivalent. 
\begin{enumerate}[(i)]
\item\label{T2.i} The cumulative maximal local expected utility is finite, i.e., 
\begin{equation*}
 \int_0^T \locutilMMV_t(\lMMV_t) \d A_t<\infty.
\end{equation*}
\item\label{T2.ii} The strategy $\lMMV$ integrates $R$.
\end{enumerate}
Furthermore, if one of these conditions holds, then the realized local utility $\utilMMV\circ (\lMMV\sint R)$ as well as the processes $(\id\wedge1)\circ (\lMMV\sint R)$ and  $(\id\wedge1)^2\circ (\lMMV\sint R)$ are special with
\begin{equation}\label{eq:Granada241018.2}
 2\locutilMMV(\lMMV)\sint A = B^{2\utilMMV\circ (\lMMV\sint R)}=B^{(\id\wedge1)\circ (\lMMV\sint R)}=B^{(\id\wedge1)^2\circ (\lMMV\sint R)} .
\end{equation}
\end{theorem}

%=================% 
\subsection{Optimal monotone mean--variance portfolios}\label{SS:2.6}

The next theorem is remarkable in that one does not require an \emph{a priori} existence of an equivalent $\sigma$-martingale measure with square-integrable density or square integrability of the yield process $R$. 

\begin{theorem}\label{T2}
If the cumulative maximal local expected utility is finite, i.e., 
$$ \int_0^T \locutilMMV_t(\lMMV_t) \d A_t< \infty,$$ 
then the maximal monotone mean--variance utility is given by
\begin{equation*}
\vMMV(x) = \sup_{\vartheta \in \AdmMMV}\, \{V_{\MMV}(x+\vartheta \sint R_T)\} 
=x+\dfrac{\Exp\bigs(-2\locutilMMV(\lMMV)\sint A\bigs)_T^{-1}-1}{2}<\infty,\quad x\in\R.
\end{equation*}
Furthermore, the optimal zero-cost strategy for expected monotone quadratic utility reads
\[
\aMMV(0) = \argmax_{\vartheta \in \AdmMMV}\{\E[\utilMMV(\vartheta \sint R_T)]\}
=\lMMV\Exp\bigs( -\left(\id\wedge1\right) \circ (\lMMV\sint R) \bigs)_{-}
\]
and the optimal monotone mean--variance investment strategy equals
\begin{equation*}
\bMMV(x) =\argmax_{\vartheta \in \AdmMMV}\, \{V_{\MMV}(x+\vartheta \sint R_T)\}= (1+2\vMMV(0))\aMMV(0). 
\end{equation*}
\end{theorem}

The complementary case gives rise to near-arbitrage opportunities. Observe that Theorems~\ref{T2} and \ref{T3} combined together yield that $\vMMV(0)$ is finite if and only if the cumulative maximal local expected utility is finite.
\begin{theorem}\label{T3}
If the cumulative maximal local expected utility is infinite, i.e., 
$$ \int_0^T \locutilMMV_t(\lMMV_t) \d A_t=\infty,$$ 
then there is a sequence of admissible strategies whose negative terminal wealth goes to $0$ in $L^{2}$ and whose positive terminal wealth is greater than or equal to $1$ with probability approaching $1$. Furthermore, there is a sequence of admissible strategies whose monotone mean--variance utility goes to $\infty$, i.e., $\vMMV(0)=\infty$.
\end{theorem}

\begin{remark}
Let us recast the optimal strategy in an alternative form. Fix $x\in\R$, $\gamma>0$, and recall  from  Remark~\ref{R:risk aversion 1} the role of $\gamma$ in shaping the optimal strategy: the optimal investment $\widetilde \beta(x)$ for MMV utility with risk aversion $\gamma$ coincides with the optimal risky investment for expected monotone quadratic utility with risk aversion $\overline\gamma = \gamma/(1+2\vMMV(0))$ and initial wealth $0$. The local risk tolerance  of this utility function reads
$\utilMMV'(\overline{\gamma}\id)/\overline\gamma=(1/\overline{\gamma} - \id)^+$. 
Since $\lMMV$ is the optimal dollar investment per unit of local risk tolerance,  we obtain
\[\widetilde \beta(x) = \bigg(\frac{1}{\overline\gamma}-\big(\widetilde{W}_--x\big)\bigg)^+\lMMV,\]
where $\widetilde{W} = x + \widetilde\beta(x)\sint R$ is the value of the optimal investement for the MMV utility with risk aversion $\gamma$.
\end{remark}

The proofs of both theorems are found in Appendix~\ref{appx:F}. We also state analogous results for the classical mean--variance preferences. These will be used to illustrate the difference between the classical and monotone mean--variance portfolio allocation in Example~\ref{E:2} and are new to the extent that they do not assume existence of an equivalent $\sigma$-martingale measure with square-integrable density or square integrability of the yield process $R$. 

Recall the mean--variance utility in \eqref{eq:MV1}, the quadratic utility function $\utilMV$ in \eqref{eq:utilMV}, and consider the corresponding optimal portfolio allocation problems 
\begin{align} 
\vMV(x) ={}& \sup_{\vartheta \in \AdmMV}\, \{V_{\MV}(x+\vartheta \sint R_T)\},\quad\qquad x\in\R, \notag \\
\uMV(x) ={}& \sup_{\vartheta \in \AdmMV}\{\E[\utilMV(x+\vartheta \sint R_T)]\},\qquad\, x\in\R.    \label{eq:uMV}
\end{align}
Denote the corresponding optimizers, should they exist, by $\bMV(x)$ and $\aMV(x)$, respectively. Observe that the expected utility maximization problem \eqref{eq:uMV} is time consistent. 
\begin{theorem}\label{T4}
There is a deterministic predictable process $\lMV$ that maximizes, for each $t\in[0,T]$, the local expected quadratic utility $\locutilMV:[0,T]\times\R^d\to [-\infty,\infty)$ given by 
\begin{equation*}
 \locutilMV_t(\lambda) = b_t^{ R[1]}\lambda - \frac{1}{2}\lambda c_t^{ R}\lambda^\top
+\int_{\R^d} \bigs(\utilMV(\lambda x)-\lambda h(x)\bigs)F_t^{ R}(\d x).
\end{equation*}
We have $2\uMV(0)<1$ if and only if $\vMV(0)<\infty$ if and only if $\int_0^T \locutilMV_t(\lMV_t)\d A_t<\infty$, in which case $\lMV$ integrates $R$, 
\begin{equation}\label{eq:Metabief250114}
2\locutilMV(\lMV)\sint A = B^{2\utilMV\circ (\lMV\sint R)}=B^{\lMV\sint R}=B^{\id^2\circ (\lMV\sint R)},
\end{equation}
the maximal  expected quadratic utility with zero initial endowment satisfies 
\[2\uMV(0) = 1-\Exp\bigs(-2\locutilMV(\lMV)\sint A\bigs)_T,\]
and the corresponding optimal investment strategy reads
\begin{equation*}
\aMV(0) = \argmax_{\vartheta \in \AdmMV}\{\E[\utilMV(\vartheta \sint R_T)]\}
=\lMV\Exp\bigs( -\lMV\sint R \bigs)_{-}.
\end{equation*}
Furthermore, the maximal mean--variance utility is then given by
\begin{equation*}
\vMV(x) =x+\dfrac{\Exp\bigs(-2\locutilMV(\lMV)\sint A\bigs)_T^{-1}-1}{2}<\infty,\quad x\in\R,
\end{equation*}
and the optimal mean--variance utility investment strategy equals
\begin{equation*}
\bMV(x) =\argmax_{\vartheta \in \AdmMMV}\, \{V_{\MV}(x+\vartheta \sint R_T)\}= (1+2\vMV(0))\aMV(0). 
\end{equation*}
\end{theorem}
We omit the proof since it is entirely analogous to the proofs of Propositions~\ref{P1} and \ref{P:Cierne240808} and Theorems~\ref{T1}--\ref{T3}, with $\AdmMMV$ in \eqref{eq:250111.MMV} replaced by $\AdmMV$ in \eqref{eq:250111.MV}; Proposition~\ref{P:London241109} replaced by Proposition~\ref{P:Cierne250106}; $\utilMMV'-1=\id\wedge1$ replaced by $\utilMV'-1=\id$; etc.
%=================% 
\subsection{Duality}\label{SS:2.7}
Since we do not assume square integrability of the yield process $R$, the appropriate dual concept is that of a separating measure; cf.~\citet*{bellini.frittelli.02}. 
\begin{definition}\label{D:separating}
We say that $\Q\ll\P$ with $\frac{\d\Q}{\d\P}\in L^2(\P)$ is a separating measure if for all $\vartheta\in\AdmMMV$, one has $\vartheta\sint R_T\in L^1(\Q)$  and $\E^{\Q}[\vartheta\sint R_T]\leq0$.
\end{definition}

By standard duality arguments, a separating measure imposes an upper bound on the maximal monotone mean--variance utility. From the  proof of the next proposition, it follows directly that $ \Var\left(\frac{\d\Q}{\d\P}\right)$ is twice the maximal MMV utility in a complete market with prices determined by $\Q$ for an agent with initial wealth $0$. The proofs of this and other statements of this subsection are found in Appendix~\ref{appx:G}.
\begin{proposition}\label{P:Granada241017}
Any separating measure $\Q$ for $R$ satifies 
$2\vMMV(0)\leq \Var\left(\frac{\d\Q}{\d\P}\right)$.
\end{proposition}

The marginal utility of the optimal portfolio now yields the separating measure whose density has the smallest variance. 
\begin{theorem}\label{T5}
Assuming $\vMMV(0)<\infty$, the following statements hold.
\begin{enumerate}[(1)]
\item The process $\frac{\Exp( -\left(\id\wedge1\right) \circ(\lMMV\sint R))}{\Exp(B^{-\left(\id\wedge1\right) \circ(\lMMV\sint R)})}$ is a well-defined square-integrable martingale.
\item\label{T5.2} The probability measure $\QMMV$ defined by 
\begin{equation*}
 \frac{\d \QMMV}{\d \P} = \frac{\Exp\bigs( -\left(\id\wedge1\right) \circ(\lMMV\sint R) \bigs)_T}{\Exp\bigs(B^{-\left(\id\wedge1\right) \circ(\lMMV\sint R)}\bigs)_T}= \frac{\utilMMV'(\aMMV(0)\sint R_T)}{\E[\utilMMV'(\aMMV(0)\sint R_T)]},
\end{equation*}
is a separating measure.

\item\label{T5.3} The maximal MMV utility in the complete market given by $\QMMV$ equals $\vMMV(0)$, i.e.,
\[2\vMMV(0)=\Var\bigg(\frac{\d\QMMV}{\d\P}\bigg).\]  
\item The density of $\QMMV$ with respect to $\P$ has the smallest variance among all separating measures and the measure $\QMMV$ is determined uniquely, unlike $\lMMV$.
\item $\QMMV$ preserves the independent increments property (resp., the L\'evy property) of the yield process $R$ and, provided $\Delta R>-1$, of the log return process $\log\Exp(R)$. 
\end{enumerate}
\end{theorem}

We next formulate necessary and sufficient conditions for $\QMMV$ to be a martingale measure for $R$. In fact, the correct concept for $R$ with unbounded jumps is a $\sigma$-martingale measure. If this seems too abstract, Proposition~\ref{P:Metabief250115} shows that in the most frequently encountered models the $\sigma$-martingale  property already yields that $\QMMV$ is a martingale measure for $R$. We also provide necessary and sufficient conditions for $\QMMV$ to be equivalent. Observe that $\QMMV$ being equivalent to $\P$ implies absence of arbitrage over trading strategies in $\AdmMMV$ even if $\QMMV$ fails to be a $\sigma$-martingale measure. 

\begin{theorem}\label{T6}
Assuming $\vMMV(0)<\infty$, the following are equivalent.
\begin{enumerate}[(i)]
\item\label{T6.i} $\QMMV$ is a $\sigma$-martingale measure for $R$.
\item\label{T6.ii} $\id\mkern1mu\utilMMV'(\lMMV\mkern1mu\id)\circ R$ is $\sigma$-special and $b^{\id\mkern1mu\utilMMV'(\lMMV\mkern1mu\id)\circ R}=0$, $A$-almost everywhere on $[0,T]$.
\end{enumerate}
Furthermore, $\QMMV$ is equivalent to $\P$ if and only if $\lMMV\Delta R< 1$ holds  $(\P\otimes A)$-almost everywhere. 
\end{theorem}
\noindent Example~\ref{E:4} shows that $\QMMV$ may be equivalent to $\P$ without being a $\sigma$-martingale measure.

\begin{proposition}\label{P:Metabief250115}
Suppose that $\QMMV$ is a $\sigma$-martingale measure for $R$. Then $\vartheta\sint R$ is a $\QMMV$-martingale for every deterministic $\vartheta\in\AdmMMV$. In particular, if the jumps of $R$ are bounded from below, then $R$ is a $\QMMV$-martingale and $\QMMV$ is a martingale measure for $R$.
\end{proposition}

 Recall the optimal local expected quadratic utility strategy $\lMV$ from Theorem~\ref{T4}. The next theorem summarizes results analogous to Proposition~\ref{P:Granada241017} and Theorems~\ref{T5} and \ref{T6} for the classical quadratic utility. For a signed measure $\Q$ and  $W\in L^1(\Q) = L^1(\Q^+)\cap L^1(\Q^-)$, we write $\E^{\Q}[W] = \int W\d\Q^+ -\int W\d\Q^- $. Here $\Q^+$ and $\Q^-$ denote the positive and negative variations of $\Q$, respectively; see \citet[VI.6.6]{rudin.87}.
\begin{theorem}\label{T7}
Suppose $\vMV(0)<\infty$. Then $\Exp\bigs(B^{-\lMV\sint R}\bigs)>0$ and the (signed) measure $\QMV$, given by 
\begin{equation}\label{eq:QMV}
 \frac{\d \QMV}{\d \P} = \frac{\Exp\bigs( -\lMV\sint R \bigs)_T}{\Exp\bigs(B^{-\lMV\sint R}\bigs)_T}
=\frac{\utilMV'(\aMV(0)\sint R_T)}{\E[\utilMV'(\aMV(0)\sint R_T)]},
\end{equation}
is a separating measure over $\AdmMV$, its density has the smallest variance among all separating (signed) measures over $\AdmMV$, and satisfies $2\vMV(0)=\Var(\frac{\d\QMV}{\d\P})$. Furthermore, if $\QMV$ is a probability measure, the following are equivalent.
\begin{enumerate}[(i)]
\item $\QMV$ is a $\sigma$-martingale measure for $R$.
\item $\id\mkern1mu\utilMV'(\lMV\mkern1mu\id)\circ R$ is $\sigma$-special and $b^{\id\mkern1mu\utilMV'(\lMV\mkern1mu\id)\circ R}=0$, $A$-almost everywhere on $[0,T]$.
\end{enumerate}
\end{theorem}
 We shall call $\QMV$ the \emph{variance-optimal separating measure} since for square-integrable $R$, the measure $\QMV$ coincides with the variance-optimal $\AdmMV$-martingale measure of \citet{schweizer.96}. We next characterize the circumstances under which the simpler MV optimizer already yields an optimal MMV portfolio.
\begin{theorem}\label{T8}
If $\vMV(0)<\infty$ and $R$ is $\sigma$-locally square integrable (i.e., $\id^2\circ R$ is $\sigma$-special),  then the following are equivalent. 
\begin{enumerate}[(i)]
\item\label{T8.i} $\lMV\Delta R\leq 1$ holds  $(\P\otimes A)$-almost everywhere.
\item\label{T8.ii}  $\lMV$ is an optimizer of the local expected monotone quadratic utility $\locutil$ in \eqref{eq:locutil}.
\item\label{T8.iii}   $\vMMV(0)=\vMV(0)$.
\item\label{T8.iv} $\vMMV(0)$ is finite and $\QMMV=\QMV$.
\item\label{T8.v} $\QMV$ is a probability measure.
\end{enumerate}
Furthermore, if any of these conditions holds, then $\QMV$ is a $\sigma$-martingale measure for $R$.
\end{theorem}

Example~\ref{E:3} shows the theorem is false without the assumption of $\sigma$-local square integrability of $R$. 

%=================%
%=================%
\section{Economic implications}\label{S:3}
This section is concerned with economic insights contained in Theorems~\ref{T1}--\ref{T3}, \ref{T4}, \ref{T5}, and \ref{T6}. For ease of exposition, we begin with the mean--variance analysis in Theorem~\ref{T4}.
\subsection{Classical mean--variance analysis}\label{SS:3.1}
Consider first a single-period model with a single return $R\in L^2$. Simple computations yield
\begin{align}
2\uMV(0) ={}& 1-\min_{\alpha\in\R}\{\E[(1-\alpha R)^2]\} = \frac{(\E[R])^2}{\E[R^2]} = \HR^2_R,\label{eq:Metabief250116}\\
2\vMV(0) ={}& \max_{\beta\in\R}\{2\E[\beta R]-\Var(\beta R)\} = \frac{(\E[R])^2}{\Var(R)}=\SR^2_R.\notag
\end{align}
Here the Hansen ratio $\HR_R$ is the ratio of mean to $L^2$ norm, while the Sharpe ratio $\SR_R$ is the ratio of mean to standard deviation, with the convention $0/0=0$. Denoting by $\aMMV$ the maximizer in \eqref{eq:Metabief250116}, note that there are several equivalent expressions for the squared Hansen ratio
\[\HR^2_R =  \E[\aMMV R] = \E[(\aMMV R)^2] = \E[2\utilMV(\aMMV R)].\]
We have encountered this phenomenon previously in \eqref{eq:Metabief250114}.

Observe that, regardless of optimality, for a random variable $W\in L^2$ one also has 
\begin{equation}\label{eq:SRHR}
 1+\SR_W^2 = \frac{1}{1-\HR^2_W}.
\end{equation}
Hence, (1) a portfolio with the highest Hansen ratio also has the highest Sharpe ratio and vice versa; (2) in the two optimization problems above, one has 
\[ 1+2\vMV(0) = \frac{1}{1-2\uMV(0)}. \]

Consider now a continuous-time model with time horizon $T$ and an arbitrary but fixed admissible strategy $\vartheta\in\AdmMV$ with terminal wealth $\vartheta\sint R_T$. By restricting ourselves to trading only multiples of $\vartheta\sint R_T$, we are now in the one-period setting considered at the beginning of this subsection. This yields
\begin{equation}\label{eq:250113}
\E[2\utilMV(\vartheta\sint R_T)]\leq \max_{\alpha\in\R}\E[2\utilMV(\alpha\vartheta\sint R_T)] = \HR^2_{\vartheta\sint R_T},
\end{equation}
i.e., twice the expected quadratic utility of a given strategy provides a lower bound on the squared Hansen ratio of the strategy.

Taking this reasoning one step further in continuous-time models, we conclude that 
\[ 2\uMV(0)=\max_{\vartheta\in\AdmMV}\E[2\utilMV(\vartheta\sint R_T)] = \HR^2_{\aMV(0)\sint R_T}\] 
yields the highest squared Hansen ratio among all admissible portfolios with zero initial cost. The portfolio $\aMV(0)$ is therefore MV efficient among all zero-cost portfolios.

This reasoning can now be applied locally in time. The expression for maximal expected quadratic utility in Theorem~\ref{T4} translates to 
\begin{equation}\label{eq:Cierne250102a}
1-2\uMV(0) = \Exp(-\widetilde K)_T,
\end{equation}
where we may interpret $\widetilde K=B^{2\utilMV\circ (\lMV\sint R)}$ as the cumulative maximal local squared Hansen ratio. Thus \eqref{eq:Cierne250102a} says that 
\begin{quote}
\emph{the maximal global squared Hansen ratio $2\uMV(0)$ equals the nominal discount rate obtained by continuously discounting at the maximal local squared Hansen ratio rate $\d \tilde K$.}
\end{quote}

In view of $\SR^2 = \frac{\HR^2}{1-\HR^2}$, we may interpret $\widehat K = \frac{1}{1- \Delta \widetilde K}\sint \widetilde K$ as the cumulative maximal local squared Sharpe ratio. By the Yor formula, one has $\Exp(-\widetilde K)\Exp(\widehat K)=1$, which gives
\[1+2\vMV(0) = \frac{1}{1-2\uMV(0)} = \frac{1}{\Exp(-\widetilde K)_T} = \Exp(\widehat K). \]
 This yields that 
\begin{quote}
\emph{the maximal global squared Sharpe ratio $2\vMV(0)$ equals the nominal rate of return obtained by continuously compounding at the maximal local squared Sharpe ratio rate $\d\widehat K$.}
\end{quote}

In models without fixed jump times, the local squared Hansen ratio and the local squared Sharpe ratio coincide, i.e., $\widetilde K = \widehat K$.

\subsection{Monotone mean--variance analysis}
\emph{Mutatis mutandis}, everything we have seen in the previous subsection translates to monotone mean--variance analysis. For a random variable $W\in L^0_+- L^2_+$, define the monotone Hansen ratio $\MHR_W$ by
\[\MHR_W = \sup_{Y\geq0}\HR_{W-Y}\]
and the monotone Sharpe ratio by
\[\MSR_W = \sup_{Y\geq0}\SR_{W-Y}.\]
Here and below, the suprema are taken over $Y\geq0$ such that $W-Y\in L^2$.  Since the function $(-1,1)\ni x\mapsto x/\sqrt{1-x^2}$ is monotone increasing, the link between the two quantities is provided by
\begin{align*}
\MSR_W ={}& \sup_{Y\geq0}\SR_{W-Y} = \sup_{Y\geq0}\frac{\HR_{W-Y}}{\sqrt{1-\HR^2_{W-Y}}} = \frac{\MHR_{W}}{\sqrt{1-\MHR^2_{W}}},
\end{align*}
which mirrors \eqref{eq:SRHR} in that 
\begin{equation}\label{eq:MSRMHR} 
1+\MSR_W^2 = \frac{1}{1-\MHR^2_W}.
\end{equation}
However, unlike the classical MV analysis, for $W\in L^2$ we cannot expect $\MSR_{-W}=-\MSR_W$ to hold.  

By the arguments in the previous subsection, for $\E[W]\in(0,\infty]$ with $W^-\neq 0$ we have
\begin{align}\label{eq:MHR2}
\begin{split}
\MHR^2_W = {}&\sup_{Y\geq0}\max_{\alpha\geq 0}\E[2\utilMV(\alpha(W-Y))] \\
         = {}&\max_{\alpha\geq 0}\sup_{Y\geq0}\E[2\utilMV(\alpha(W-Y))] \\
         = {}&\max_{\alpha\geq 0}\E[2\utilMMV(\alpha W)].
\end{split}				
\end{align}
On the second line, the optimal $Y$ for given $\alpha>0$ is one that satisfies $\alpha Y = (\alpha W -1)\vee 0$, indeed, $Y = (W-1/\alpha)\vee 0$; see also Figure~\ref{F:1}.

Observe that for maximization of $\MHR_W$ over some subset of $L^0_+- L^2_+$ containing $0$, one can disregard random variables $W$ with negative mean since their MHR is non-positive. For such $W\in L^2$ one does not obtain a nice formula for MHR like \eqref{eq:MHR2} but it is not difficult to see that $\HR(W-n\indicator{A})\to -\sqrt{\P[A]}$ as $n\uparrow\infty$, hence for diffuse $W\in L^2$ with negative mean one has $\MHR_W=0$.

We conclude from \eqref{eq:MHR2} that $\aMMV(0)$ has the highest monotone Hansen ratio, hence also the highest monotone Sharpe ratio thanks to conversion \eqref{eq:MSRMHR}, among zero-cost strategies in $\AdmMMV$. The optimal amount to set aside is $Y = (\aMMV(0)\sint R_T -1)\vee 0$, thus by \eqref{eq:250113} we also have 
\[2\uMMV(0) = \HR^2_{(\aMMV(0)\sint R_T)\wedge1} = \E[2\utilMV((\aMMV(0)\sint R_T)\wedge1)].\]

Consider now an arbitrary but fixed admissible strategy $\vartheta\in\AdmMMV$ with terminal wealth $\vartheta\sint R_T$ such that $\E[\vartheta\sint R_T]\in(0,\infty]$ and $(\vartheta\sint R_T)^-\neq0$. By restricting ourselves to trading only multiples of $\vartheta\sint R_T$, we obtain
\[\E[2\utilMMV(\vartheta\sint R_T)]\leq \max_{\alpha\in\R}\E[2\utilMMV(\alpha\vartheta\sint R_T)] = \MHR^2_{\vartheta\sint R_T},\]
i.e., twice the expected monotone quadratic utility of a given strategy with positive (or infinite) mean provides a lower bound on the squared monotone Hansen ratio of the strategy. These observations are illustrated in Example~\ref{E:2}.

Finally, we briefly summarise the link between local and global utility in the monotone case, which too is completely analogous to the previous subsection. The expression for maximal expected monotone quadratic utility from Proposition~\ref{P1} and Theorem~\ref{T2} translates to 
\begin{equation*}
1-2\uMMV(0) = \Exp(-\widetilde K^{\MMV})_T,
\end{equation*}
where we may interpret $\widetilde K^{\MMV}=B^{2\utilMMV\circ (\lMMV\sint R)}$ as the cumulative maximal local squared MHR. 
Using the conversion between the (monotone) Hansen and the (monotone) Sharpe ratio, we also obtain
\[1+2\vMMV(0) = \Exp(\widehat K^{\MMV})_T, \]
where $\widehat K^{\MMV} = \frac{1}{1- \Delta \widetilde K^{\MMV}}\sint \widetilde K^{\MMV}$ is the cumulative local squared monotone Sharpe ratio.

%=================%
%=================%
\section{Examples}\label{S:4}
Example~\ref{E:1} illustrates that the locally optimal strategy $\lMMV$ may not be unique even if there is no redundancy among the traded assets. Example~\ref{E:2} contrasts classical mean--variance and monotone mean--variance optimal portfolios in a continuous-time setting. Example~\ref{E:3} shows that the equivalences in Theorem~\ref{T8} fail when $R$ is not ($\sigma$-locally) square integrable. Example~\ref{E:4} illustrates that $\QMMV$ can be equivalent to $\P$ without being a $\sigma$-martingale measure for $R$, again highlighting the importance of the ($\sigma$-local) square integrability assumption in Theorem~\ref{T8}.  Example~\ref{E:5} constructs a model with square-integrable yield $R$, where the maximal Sharpe ratio available by trading is finite but the maximal monotone Sharpe ratio is infinite. Example~\ref{E:6} exhibits a complete market model with an equivalent martingale measure, where both the MV utility and the MMV utility are infinite. 

\begin{example}\label{E:1}
 Consider a one-period model with four atoms and two assets, whose yield distribution is given by
\begin{center}
\begin{tabular}{ccccc}
\hline & $\omega _{1}$ & $\omega _{2}$ & $\omega _{3}$ & $\omega _{4}$ \\ 
$\P[\{\omega \}]$ & 0.2 & 0.6 & 0.1 & 0.1 \\ 
$\Delta R^{1}(\omega )$ & -0.5 & 0.5 & 1 & 1.2 \\ 
$\Delta R^{2}(\omega )$ & -0.5 & 0.5 & 1.2 & 1\\[0.1ex]
\hline
\end{tabular}\smallskip
\end{center}
Choosing $q = (0.62,0.18,0.1,0.1)$ yields an equivalent martingale measure $\Q$ for $R$, hence the model is free of arbitrage. Observe that in this model, all separating measures are martingale measures for $R$.
One easily verifies that $\hat q = (0.5,0.5,0,0)$ defines a martingale measure $\QMMV$ whose density satisfies
$\Var(\frac{\d \QMMV}{\d\P})=\frac{2}{3}$.
The duality constraint in Proposition~\ref{P:Granada241017} and \eqref{eq:vMMVx} yield that twice the maximal expected quadratic utility cannot exceed 
$1-(1+\frac{2}{3})^{-1} = 0.4$. On the other hand, one has 
$$\E[2\utilMMV(R^1)] = \E[2\utilMMV(R^2)] = 0.4,$$ 
hence we have two distinct optimal strategies, each investing one dollar in one of the assets, respectively, and borrowing one dollar at the zero risk-free rate. Any convex combination of these two strategies is again optimal. The wealth of all optimal strategies coincides below the bliss point of the monotone quadratic utility function $\utilMMV$. 

In this example, 2/3 is the maximal squared monotone Sharpe ratio among zero-cost strategies and 0.4 (twice the maximal expected utility) is the maximal squared monotone Hansen ratio among zero-cost strategies.   
\end{example}

The calculations in the next three examples are easily implemented using the drift formula \eqref{eq:Cierne251205}, bearing in mind that 
$g\circ (\lambda\sint R) = g(\lambda\mkern2mu\id)\circ R = g(\lambda(\e^\id-1))\circ X$
by composition (Proposition~\ref{P:composition}) for all suitably regular $g$ and all $\lambda\in L(R)$. In the examples, we specify the numerical value of $b^{X[0]}$ with the conversion $b^{X[1]}=b^{X[0]}+\int_{[-1,1]}xF^X(\d x)$. The drift formula \eqref{eq:Cierne251205} in one dimension then simplifies to 
\[
b^{\xi\circ X} = \xi'(0)b^{X[0]}+\frac{1}{2}\xi''(0)c^{X} + \int_{\R} \xi(x)F^X(\d x).
\]

\begin{example}\label{E:2}
Let $X$ be a L\'evy process with characteristics $(b^{X[0]}=\mu,\sigma^2,F^X= N(0,\gamma^2))$ relative to the activity process $A_t=t$, where $N(\cdot,\cdot)$ denotes the cumulative normal distribution with a given mean and variance, respectively. We will use the specific numerical values  $\mu=0.2$, $\sigma=0.2$, and $\gamma = 0.1$, which are broadly consistent with the empirical distribution of the logarithmic returns of a well-performing stock. We shall also fix a one-year time horizon, $T=1$.

Consider a time-homogenous model with a single risky asset, whose yield is given by $R = (\e^{\id}-1)\circ X$. Theorem~\ref{T4} yields 
\begin{align*}
\lMV ={}& \frac{b^{\id\circ R}}{b^{\id^2\circ\mkern-1mu R}} = \frac{b^{(\e^{\id}-1)\circ\mkern-1mu X}}{b^{(\e^{\id}-1)^2\circ\mkern-1mu X}} 
= \frac{\mu+\sigma^2/2+\e^{\gamma^2/2}-1}{\sigma^2+\e^{2\gamma^2}-2\e^{\gamma^2/2}+1}
\approx 4.4844, \\
b^{2\utilMV\circ (\lMV\sint R)} ={}& \frac{(b^{\id\circ R})^2}{b^{\id^2\circ\mkern-1mu R}} = 
 \frac{(\mu+\sigma^2/2+\e^{\gamma^2/2}-1)^2}{\sigma^2+\e^{2\gamma^2}-2\e^{\gamma^2/2}+1} \approx 1.0091,\\
2\vMV(0) ={}&\exp\bigs(b^{2\utilMV\circ (\lMV\sint R)}T\bigs)-1\approx 1.7430.
\end{align*}
These calculations also follow from the classical mean--variance utility theory; see, for example, \citet*[Proposition~3.6, Lemma~3.7, Corollary~3.20, and Proposition~3.28]{cerny.kallsen.07}. The optimal strategy for the expected quadratic utility reads
\begin{equation*}
\aMV(0) = \lMV\Exp\bigs( -\lMV\sint R \bigs)_{-} 
\end{equation*}
with wealth
\begin{equation}\label{eq:Granada241119}
\aMV(0)\sint R_T = 1-\Exp\bigs( -\lMV\sint R \bigs)_T. 
\end{equation}

The quadratic utility $\utilMV$ attains its maximum at 1. Setting aside all the wealth above 1 therefore increases the expected quadratic utility and hence the Hansen ratio of the resulting wealth
\begin{equation}\label{eq:Granada241120}
(\aMV(0)\sint R_T)\wedge 1 = 1-\Exp\bigs( -\lMV\sint R \bigs)_T\vee0.
\end{equation}
To explore the change from \eqref{eq:Granada241119} to \eqref{eq:Granada241120}, observe that the stochastic exponential $\Exp( -\lMV\sint R )$ switches sign with intensity $\theta^{\MV} = \int_{\R} \indicator{ \lMV(\e^x-1)>1}F^X(\d x)\approx 2.2057\%$. Using the properties of the Poisson distribution, the probability of switching exactly once is $\exp({-\theta^{\MV} T})\theta^{\MV}T\approx 2.1575\%$ while the probability of switching an odd number of times is $\P[\Exp\bigs( -\lMV\sint R \bigs)<0]=\frac{1-\exp({-2\theta^{\MV}})}{2} \approx 2.1577\%$. Thus the first switch captures essentially all the probability that the wealth $\aMV(0)\sint R_T$ will finish above the bliss point.
 
We shall now compute the first two moments of the truncated wealth $(\aMV(0)\sint R_T)\wedge 1$ with the aim of subsequently obtaining its Hansen and Sharpe ratio, respectively. For $p\in\{0,1,2\}$, let $\phi^{\MV}_+(p) =\E[|\Exp(-\lMV\sint R)_T|^p\indicator{\{\Exp(-\lMV\sint R)_T>0\}}]$ and $\phi^{\MV}_-(p) =\E[|\Exp(-\lMV\sint R)_T|^p\indicator{\{\Exp(-\lMV\sint R)_T<0\}}]$. Using the Mellin transform technique for signed stochastic exponentials developed in \citet*[Example~4.4]{cerny.ruf.23.spa}, we have
\begin{align*}
2\phi^{\MV}_{\pm}(p) ={}& \exp\big(b^{(|1-\id|^p\indicator{\id\neq1} -1)\circ (\lMV\sint R)}T\big)
\pm\exp\big(b^{(|1-\id|^p(\indicator{\id<1}-\indicator{\id>1})-1)\circ (\lMV\sint R)}T\big). 
\end{align*}

Observe that $\P[\Exp(-\lMV\sint R)_T<0] = \phi^{\MV}_-(0)$. We further have 
\begin{align*}
\E[(\aMV(0)\sint R_T)\wedge 1 ] ={}& 1-\phi^{\MV}_+(1) \approx 0.63373,\\
\E[((\aMV(0)\sint R_T)\wedge 1)^2 ] ={}& 1-2\phi^{\MV}_+(1)+\phi^{\MV}_+(2) \approx 0.63136,\\
\E[\aMV(0)\sint R_T] - \E[(\aMV(0)\sint R_T)\wedge 1 ] ={}&  \phi^{\MV}_-(1)\approx 0.0017,\\
\E[(\aMV(0)\sint R_T)^2] - \E[((\aMV(0)\sint R_T)\wedge 1)^2 ] ={}&  2\phi^{\MV}_-(1)+\phi^{\MV}_-(2)\approx 0.0041.
\end{align*}
The mean and the second moment of $\aMV(0)\sint R_T$ are equal to $1-\exp\bigs(-b^{2\utilMV\circ (\lMV\sint R)}T\bigs) \approx 0.6354$. The ``setting aside'' of gains above $1$ produces an increase in the squared Hansen ratio from $\HR^2_{\aMV(0)\sint R_T} =0.6354$ to $\HR^2_{(\aMV(0)\sint R_T)\wedge1} = 0.6361$. This difference, here $6.78\times 10^{-4}$,  is roughly equal to $\phi^{\MV}_-(2)\approx 6.69\times 10^{-4}$. The squared Sharpe ratio is related to the squared Hansen ratio via $\SR^2 = \frac{\HR^2}{1-\HR^2}$. This quantity increases from $\exp\bigs(b^{2\utilMV\circ (\lMV\sint R)}T\bigs)-1 \approx 1.7430$ to approximately $1.7481$. The difference, here approximately $5.11\times 10^{-3}$, is closely matched by $(1+2\vMV(0))^2\phi^{\MV}_-(2)\approx 5.03\times 10^{-3}$.

Let us now consider the optimal expected monotone quadratic utility strategy 
$$\aMMV(0) = \lMMV\Exp\bigs( -\left(\id\wedge1\right) \circ (\lMMV\sint R) \bigs)_{-}.$$ 
Numerical optimization of the local expected monotone quadratic utility yields $\lMMV \approx 4.5143$, which is slightly higher than $\lMV\approx 4.5130$. The strategy $\aMMV(0)$ stops investing in the risky asset once its wealth reaches or exceeds the bliss point 1.  The arrival intensity of jumps that satisfy $\lMMV\Delta R\geq1$ is $\theta^{\MMV} = \int_{\R} \indicator{ \lMMV(\e^x-1)\geq1}F^X(\d x)\approx 2.2699\%$. The probability of finishing at or above  the bliss point is the complement of no such jump occurring, i.e., $\P[\aMMV(0)\sint R_T\geq 1] = 1-\exp(-\theta^{\MMV}T)\approx2.244\%$. Thus the optimal expected monotone quadratic utility investment slightly increases the probability of wealth being above the bliss point, from $2.158\%$ for the classical expected quadratic utility strategy. 

The squared monotone Sharpe ratio of the optimal strategy $\aMMV(0)$, i.e., the squared Sharpe ratio of the truncated wealth $(\aMMV(0)\sint R_T)\wedge 1$, reads
\[
2\vMMV(0) ={}\exp\bigs(b^{2\utilMMV\circ (\lMMV\sint R)}T\bigs)-1\approx 1.7482.
\] 
We conclude that $(\aMMV(0)\sint R_T)\wedge 1$ outperforms $(\aMV(0)\sint R_T)\wedge 1$ very marginally in terms of the squared Sharpe ratio.
\end{example}

\begin{example}\label{E:3}
Fix $T=1$. Consider an $\R$-valued  L\'evy process $X$ with characteristics 
$$\left(b^{X[0]}=0,0,F^X(\d x)= (\e^{4x}\indicator{x<0}+\e^{-x}\indicator{x>0})\d x\right)$$
and let $R=(\e^{\id}-1)\circ X$. Then $\E[R_T]=\infty$ and  $\locutilMV(\lambda)=-\infty\indicator{\lambda\neq 0}$, hence we have $\lMV=0$ and the variance-optimal measure $\QMV =\P$ is an equivalent separating measure over $\AdmMV$ but not a $\sigma$-martingale measure (indeed, $R$ is not even $\sigma$-special). 

On the other hand, we obtain numerically $\lMMV\approx 1.108$ with $b^{\id\mkern1mu\utilMMV'(\lMMV\mkern1mu\id)\circ R}=0$. The arrival intensity of jumps that satisfy $\lMMV\Delta R\geq1$ reads 
\[\theta = \int_{\R} \indicator{ \lMMV(\e^x-1)\geq1}F^X(\d x) = \frac{\lMMV}{1+\lMMV}\approx 0.5256.\]
Thus the minimal-variance separating measure $\QMMV$ from Theorem~\ref{T5}\ref{T5.2} is a non-equivalent $\sigma$-martingale measure for $R$ with 
\[\P\bigg[\frac{\d\QMMV}{\d\P}=0\bigg] = \P[\Exp(-(\id\wedge1)\circ (\lMMV\sint R))_T=0] = 1-\e^{-\theta T}\approx 40.88\%.\] 
In this example, the conditions \ref{T8.i} and \ref{T8.v} of Theorem~\ref{T8} are satisfied but \ref{T8.ii}--\ref{T8.iv} fail.
\end{example}

\begin{example}\label{E:4}
Fix $T=1$. Consider an $\R$-valued  L\'evy process $X$ with characteristics 
$$\left(b^{X[0]}=0,0,F^X(\d x)= (10\e^{x}\indicator{x<0}+3\e^{-3x/2}\indicator{x>0})\d x\right)$$
and let $R=(\e^{\id}-1)\circ X$. Then $\E[R_T]=-1$ but $\E[(R_T\vee0)^2]=\infty$, hence $\locutilMMV(\lambda)<0 $ for $\lambda>0$ by Jensen's inequality and  $\locutilMMV(\lambda)=-\infty $ for $\lambda<0$ due to the lack of square integrability. This yields $\lMMV=0$ and $\QMMV =\P$ is an equivalent separating measure over $\AdmMMV$ but not a $\sigma$-martingale measure since $b^R=b^{(\e^{\id}-1)\circ X}=-1$. In this example, conditions \ref{T8.i}--\ref{T8.v} of Theorem~\ref{T8} are satisfied but $\QMMV$ is not a $\sigma$-martingale measure for $R$.
\end{example}

\begin{example}\label{E:5}
Fix $T=2$ and consider the sequence of fixed times $(\tau_n)_{n\in\N}$ given by $\tau_n = T-\frac{1}{n}$. Let $(X_n)_{n\in\N}$ be a sequence of independent random variables such that $\P[X_n=-\frac{1}{3n^3}]=\frac{1}{2}-\frac{1}{4n^2}$, $\P[X_n=\frac{1}{2n^2}]=\frac{1}{2}$, and $\P[X_n=1]=\frac{1}{4n^2}$. Let $R$ be constant on the intervals $[0,\tau_1)$, $[\tau_n,\tau_{n+1})$ for all $n\in\N$, and $[T,\infty)$. Suppose further $R_0=0$ and $\Delta R_{\tau_n} = X_n$ for all $n\in\N$. Then $R$ is a square-integrable semimartingale with independent increments and we may choose $A$ to be piecewise constant with $\Delta A_{\tau_n} = \frac{1}{n^2}$ for all $n\in\N$. Easy calculations yield 
\begin{align*}
b^{\id\circ R}_{\tau_n} = {}& n^2\E[\Delta R_{\tau_n}] 
= n^2\left(-\frac{1}{3n^3}\left(\frac{1}{2}-\frac{1}{4n^2}\right)+\frac{1}{4n^2}+\frac{1}{4n^2}\right) 
= \frac{1}{2}+\mathcal{O}\left(\frac{1}{n}\right),\\
b^{\id^2\circ R}_{\tau_n} = {}& n^2\E[(\Delta R_{\tau_n})^2] 
= n^2\left(\frac{1}{9n^6}\left(\frac{1}{2}-\frac{1}{4n^2}\right)+\frac{1}{8n^4}+\frac{1}{4n^2}\right) 
= \frac{1}{4}+\mathcal{O}\left(\frac{1}{n^2}\right),
\end{align*}
which gives 
\begin{align*}
\lMV_{\tau_n}= {}& \frac{b^{\id\circ R}_{\tau_n}}{b^{\id^2\circ R}_{\tau_n}} = 2+\mathcal{O}\left(\frac{1}{n}\right),\\
\HR^2_{X_n} = {}& \frac{(\E[X_n])^2}{\E[X_n^2 ]} = \frac{1}{n^2}\frac{(b^{\id\circ R}_{\tau_n})^2}{b^{\id^2\circ R}_{\tau_n}} 
= \frac{1}{n^2}+\mathcal{O}\left(\frac{1}{n^3}\right),
\end{align*}
and 
\[
\int_0^\infty 2\locutilMV_{u}(\lMV_{u})\d A_u = \sum_{n\in\N}\frac{1}{n^2}\frac{(b^{\id\circ R}_{\tau_n})^2}{b^{\id^2\circ R}_{\tau_n}}
=\sum_{n\in\N}\HR^2_{X_n}<\infty.
\]

For the optimal local expected monotone quadratic utility strategy we have 
\begin{align*}
\lMMV_{\tau_n} ={}& \frac{b^{\id\indicator{\id<1}\circ R}_{\tau_n}}{b^{\id^2\indicator{\id<1}\circ  R}_{\tau_n}} 
=\frac{\E[X_n\indicator{\{X_n<1\}} ]}{\E[X_n^2\indicator{\{X_n<1\}} ]}
=\frac{-\frac{1}{3n^3}\left(\frac{1}{2}-\frac{1}{4n^2}\right)+\frac{1}{4n^2}}{\frac{1}{9n^6}\left(\frac{1}{2}-\frac{1}{4n^2}\right)+\frac{1}{8n^4}} 
= 2n^2 +\mathcal{O}(n), \\
\MHR^2_{X_n}={}&\frac{1}{n^2}b^{(\id\wedge1)\circ (\lMMV\sint R)}_{\tau_n}  =\E[(\lMMV_{\tau_n}X_n)\wedge1] 
=\frac{\left(-\frac{1}{3n^3}\left(\frac{1}{2}-\frac{1}{4n^2}\right)+\frac{1}{4n^2}\right)^2}{\frac{1}{9n^6}\left(\frac{1}{2}-\frac{1}{4n^2}\right)+\frac{1}{8n^4}} +\frac{1}{4n^2} =\frac{1}{2}+\mathcal{O}\left(\frac{1}{n}\right), 
\end{align*}
hence 
\[
\int_0^T 2\locutilMMV_{u}(\lMMV_{u})\d A_u = \sum_{n\in\N}\frac{1}{n^2}b^{(\id\wedge1)\circ (\lMMV\sint R)}_{\tau_n}=\sum_{n\in\N}\MHR^2_{X_n} = \infty.
\] 
Theorem~\ref{T3} now yields that the monotone Sharpe ratio has no upper bound over strategies in $\AdmMMV$, i.e., $\vMMV(0)=\infty$. 

To understand this result intuitively, fix some  $N\in\N$ and observe that the optimal expected monotone quadratic utility strategy $\aMMV(0)$ is well defined up to time $\tau_N$. In the language of Section~\ref{S:3}, Theorem~\ref{T3} asserts that by trading optimally up to time $\tau_N$ we obtain the squared monotone Sharpe ratio of 
\[
\MSR^2_{\aMMV(0)\sint R_{\tau_N}} = \prod_{n=1}^N (1+\MSR^2_{X_n}) -1 \geq \sum_{n=1}^N \MSR^2_{X_n}.
\]
Since $\MSR^2_{X_n}=\vfrac{\MHR^2_{X_n}}{(1-\MHR^2_{X_n})}=1+\mathcal{O}\left(\frac{1}{n}\right)$, the cumulative squared local monotone Sharpe ratio increases without bound as $N\uparrow\infty$, hence the squared monotone Sharpe ratio of the optimal trade up to time $\tau_N$, too, is unbounded for $N\uparrow\infty$.
\end{example}

\begin{example}\label{E:6}
Fix $T$ and $(\tau_n)_{n\in\N}$ as in Example~\ref{E:5}. Let $(X_n)_{n\in\N}$ be a sequence of  independent random variables such that $\P[X_n=-\frac{n+1}{n^3+1}]=\frac{n^3}{n^3+1}$ and $\P[X_n=\frac{n^3-n}{n^3+1}]=\frac{1}{n^3+1}$. Let $R$ be piecewise constant as in Example~\ref{E:5} with  $\Delta R_{\tau_n} = X_n\indicator{n>1}$ for all $n\in\N$. Then $R$ is an instantaneously arbitrage-free square-integrable semimartingale with independent increments. We have
$\E[X_n] = -\frac{n}{n^3+1}$, $\E[X^2_n] = \frac{n^2(n+1)}{(n^3+1)^2}$, $\HR^2_{X_n} =\MHR^2_{-X_n} = \frac{1}{n+1}$, hence $\sum_{n=2}^\infty\HR^2_{X_n}=\infty$ and both the MV utility and MMV utility are infinite as illustrated in the previous example. 

We further have $\lMMV=\lMV$, with 
\[\lMV_{\tau_n} = \frac{\E[X_n]}{\E[X^2_n]}\indicator{n>1} = -\frac{n^3+1}{n(n+1)}\indicator{n>1},\]
does not integrate $R$ on [0,T] and the variance-optimal separating measure \eqref{eq:QMV} is not well defined. There is, in fact, no separating measure in the sense of Definition~\ref{D:separating} since infinite MMV utility is incompatible with the existence of such a measure by Proposition~\ref{P:Granada241017}. 

On the other hand, observe that the piecewise constant process $V$ with
\[\Delta V_{\tau_n} = \frac{n+(n^3+1)X_n}{n^2}\indicator{n>1},\]
is a martingale with independent increments satisfying $\Delta V>-1$. Then $\Exp(V)>0$ is a martingale by \cite[Theorem~4.1]{cerny.ruf.23.spa} and the measure $\Q$ defined by $\frac{\d\Q}{\d\P}=\Exp(V)_T\in L^1\setminus L^2$ is an equivalent $\sigma$-martingale measure for $R$. 
\end{example}

%=================%
%=================%
\section{Related studies}\label{S:5}
%=================%
\subsection{Monotone mean--variance preferences}
The notion of monotone MV efficiency has been formulated here for the first time, although it is implicitly present in previous works, for example, \citet*{maccheroni.al.09} and \citet*{cui.al.12}. The non-monotonicity of the Sharpe ratio is folklore, going back to at least \citet*[Table~1]{hodges.98.forc}.  Our definition of monotone mean--variance utility \eqref{eq:MMV1}, taken from \citet*{cerny.20}, exploits the notion of \emph{monotone hull} from \citet*{filipovic.kupper.07}. The definition is consistent with the original definition of \citet*{maccheroni.al.06} as seen in \cite[Theorem~2.1]{maccheroni.al.09} but the domain we consider extends beyond $L^2$. The monotone Sharpe ratio is defined and analysed in \citet*{cerny.03,cerny.20} and \citet*{zhitlukhin.19}. 

%=================%
\subsection{Optimal monotone mean--variance portfolios}
The literature on optimal dynamic portfolio allocation with monotone mean--variance utility is sparse. A few studies have observed that MV and MMV optimal strategies coincide in specific models:  \citet*{trybula.zawisza.19} look at diffusive stochastic volatility models; \citet*{shen.zou.22.sifin} consider a GBM model with constraints; \citet*{hu.shi.xu.23.sifin} examine diffusive models with random coefficients under constraints. \citet*{cerny.20} points out that such results are generic in all semimartingale models with continuous prices admitting a local martingale measure with square-integrable density. \citet*{strub.li.20} reach a subset of conclusions in \citet*{cerny.20} using a framework that is less general, in particular, assuming the reverse H\"older inequality for their equivalent martingale measure.

Three further studies, \citet*{li.guo.tian.24} and \citet*{shi.xu.24} in the context of reinsurance and \citet*{li.liang.pang.25.mor} in the context of financial markets, consider asset prices with square-integrable jumps. Here it is observed that the optimal MV and MMV strategies sometimes differ. \citet*{cerny.20} shows that this phenomenon, too, can be understood in terms of the classical quadratic hedging theory. In full generality (see, e.g., \citet*{cerny.kallsen.07}) for locally square-integrable $S$, there is a process $L>0$ with $L_->0$ and a process $a$ such that the density of the variance-optimal $\sigma$-martingale measure is given by $L\Exp(-a\sint S)/L_0$. Furthermore, the wealth of an efficient zero-cost strategy whose mean and second moment coincide is given by $1-\Exp(-a\sint S)$. It follows \cite[Theorem~5.6]{cerny.20} that the MV and MMV optimal portfolios coincide if and only if the optimal zero-cost wealth of the quadratic criterion does not exceed the bliss point which occurs if and only if $a\Delta S\leq1$. This line of reasoning has recently been revisited by \citet*{li.liang.pang.25.orl}. Observe, however, that in \citet*{cerny.20} the set of admissible strategies for the MMV preference is larger than the set of MV admissible strategies in line with general semimartingale expected utility theory, while in \citet{li.liang.pang.25.orl} the two sets are assumed to be the same.

The only paper on MMV efficient portfolios whose scope is similar to our work is \citet*{li.liang.pang.23.arxiv}. Their financial market $R$ is modelled by a L\'evy process whose jumps are compound Poisson with $\Delta R\geq -1$. They further assume that jumps of $R$ have finite 8-th moment with $b^R>0$. In contrast, we place no moment restrictions of $R$, allow for jumps at fixed times, and consider general time-inhomogeneous increments with jumps of infinite variation. Observe that the activity of such processes may not be absolutely continuous with respect to calendar time.

To the best of our knowledge, the optimal trading strategy and its link to minimax measures have never been explicitly characterized for the monotone mean--variance utility in the L\'evy setting, let alone for general semimartingales with independent increments. Our construction naturally deals with non-uniqueness of the optimizer illustrated in Example~\ref{E:1}. The construction also shows that the wealth of any two optimizers must coincide on the domain of monotonicity of $\utilMMV$, which offers another way of seeing that the minimax separating measure $\QMMV$ (equal to the rescaled marginal monotone quadratic utility at the optimum) is determined uniquely.

\subsection{Expected monotone quadratic utility}
The literature below deals exclusively with expected utility maximization. We have shown in Proposition~\ref{P1} that the maximization of the expected monotone quadratic utility produces MMV efficient portfolios. The literature on optimal expected utility portfolios for $R$ with independent increments typically considers only L\'evy processes and further assumes $\Delta R\geq -1$. The state-of-the-art results relevant to our work are those of \citet*{bender.niethammer.08} with $p=q=2$ in their notation. 

\citet*{bender.niethammer.08} assume $R$ is special (Standing Assumption (ii)), which we show is not required. In Proposition 2.6, they \emph{establish} existence of the local optimizer for a single risky asset under ``Assumption (H)'', stronger than our assumption of no instantaneous arbitrage, together with either $b^R>0$ or $R$ square integrable. Their existence result in one dimension therefore does not cover Example~\ref{E:4}.  In arbitrary dimension, \emph{assuming} existence of the local optimizer with property \ref{T6.ii} of Theorem~\ref{T6} (called ``condition $C^a_2$'' in their work), they establish existence and explicit form of the optimal trading strategy. Our description of the optimal strategy $\aMMV(0)$ is more streamlined than the corresponding expressions in  \cite[Section~6]{bender.niethammer.08}. 

Elsewhere, \citet*{nutz.12.mafi} considers expected power utility maximization for L\'evy processes but not the power needed here. \citet*{goll.rueschendorf.01} and \citet*{cawston.vostrikova.14} assume strictly convex $\util$; this rules out $\utilMMV$. Beyond L\'evy processes, \citet*{kallsen.muhle-karbe.10.aap} consider It\^o semimartingales with independent increments; \citet*{fujiwara.10} considers general semimartingales with independent increments. Neither of these papers obtains a result analogous to our Theorem~\ref{T1} on the integrability of the local optimizer.

%=================%
\subsection{Minimax measures}
Local martingale measures whose density has the smallest $L^2$-norm are treated in \citet*{jeanblanc.al.07} and \citet*{bender.niethammer.08}. Only the L\'evy case where $ R$ is special and $\Delta R>-1$ is considered. Theorem~2.4 in \cite{bender.niethammer.08} yields the implication \ref{T6.ii} $\Rightarrow$ \ref{T6.i} of our Theorem~\ref{T6} without asserting the existence of $\lMMV$. 

\citet*{bellini.frittelli.02} have a minimax theorem specifically suited to $\utilMMV$, but with admissibility restricted to wealth bounded uniformly from below. This rules out our optimizer in most continuous-times models including the Black--Scholes model.

In other developments, \citet*[Theorem~3.1]{fujiwara.miyahara.03} in one dimension and \citet*[Theorem~A]{esche.schweizer.05} in arbitrary dimension show that the minimal entropy equivalent \emph{martingale measure} (should it exist) preserves the L\'evy property of the yield process $R$. Here we show in Theorem~\ref{T5} that the minimax \emph{separating measure} $\QMMV$ (which may or may not be equivalent and/or a martingale measure) preserves the L\'evy property and, more broadly, the independent increments property of the yield process $R$.

\subsection{Classical mean--variance preferences}
Our analysis has some repercussions for the literature dealing with the classical Markowitz efficient portfolios. Here one typically assumes locally square-integrable yield and the existence of an equivalent $\sigma$-martingale measure with square-integrable density (e.g., \citet*{cerny.kallsen.07}). Theorem~\ref{T4} is new in that it does not assume either condition and instead deduces the optimal solution from the local expected quadratic utility maximization under the absence of instantaneous arbitrage. In our setting, the (signed) variance-optimal measure $\QMV$ exists if and only if the mean--variance utility $\vMV(0)$ is finite, however it is not always a $\sigma$-martingale measure but only a separating measure as seen in Example~\ref{E:3}.

The processes $\widehat K$ and $\widetilde K$ in Subsection~\ref{SS:3.1} are well-defined in our setting when $\vMV(0)<\infty$. They coincide with the eponymous objects in \citet[Lemma~1]{schweizer.94} when $R$ is square integrable. The intepretation of $\widetilde K$ as the cumulative maximal local squared Hansen ratio is new. The role of the Hansen ratio in quadratic portfolio optimization has previously been highlighted in \citet{cerny.20.arxiv}.

\section{Conclusions}\label{S:6}
The paper tackles optimal monotone mean--variance portfolio selection with a view of obtaining explicit formulae in the widest possible class of models in order to draw conclusions that are not skewed by the specificity of the setup. Since the maximization of the MMV utility itself is not easily phrased as a time-consistent optimization problem, we exploit the fact that MV and MMV preferences are translation-invariant hulls of expected utility. This link has previously been missed in the MMV literature (and largely also in the MV literature), leading to correspondingly more involved analyses. This point of view also yields clean economic interpretation that we bring to light by consistently applying the language of predictable variations.

The construction of the optimal strategy is not performed through some compactness principle  as in general semimartingale settings with arbitrary utility functions (cf.~\citet*[Theorem~5.4]{biagini.cerny.20}) but locally. Our construction is  explicit whereby the initial inputs (the characteristics of the yield process) are transformed into the local expected utility by a universal formula for predictable variations that can subsequently be optimized by standard convex numerical procedures. The new approach allows us to obtain a solution under the minimal assumption of no instantaneous arbitrage, which is weaker and easier to verify than the classical existence of a $\sigma$-martingale measure. Indeed the condition is weaker than the ``no unbounded profit with bounded risk'', as seen in an example with $A_t=t$, $b^R_t=t^{-1/2}\indicator{t>0}$, $c^R_t=1$, and $F_t^R=0$, for all $t\geq0$, which satisfies \eqref{eq:instNA_PII} but allows unbounded profits with bounded risk since $(b^R)^2/c^R$ is not integrable; see \citet{karatzas.kardaras.07}. Furthermore, we are able to consider completely general asset returns with independent increments without any moment restrictions. The finiteness of MMV utility does not need to be assumed \emph{ex ante}; it is checked  by looking at the local expected utility of the local optimizer that we show exists regardless.

To accomplish this research programme, we make two broader contributions to the theory of stochastic processes: (A) we consider $\sigma$-special processes --- a wider class of processes to which a drift rate can be assigned; (B) we formulate a necessary and sufficient integrability criterion for (multivariate) semimartingale stochastic integrals in terms of predictable linear--quadratic variations. Apart from being useful in our construction, item (A) brings further benefit in simplifying drift calculations for predictable variations (Lemma~\ref{L:240130}) since it removes the need for compatible truncations (cf.~\citet*[Proposition~5.2]{cerny.ruf.22.ejp}). The novel formulation in item (B) allows us to link local expected utility maximization with integrability; this link is made explicit here for the first time but we expect this observation to be useful in the broader literature on local expected utility maximization. 

In Section~\ref{S:3} we have provided an economic interpretation of the MMV utility in terms of the monotone Sharpe ratio and we have shown how the global MMV utility arises from the continuous compounding of local squared monotone Sharpe ratios. The monotone Sharpe ratio further determines the slope of the MMV efficient frontier, as argued in Remark~\ref{R:risk aversion 1}.

\appendix

%=================%
%=================%

\section{First steps}\label{appx:A}
\begin{proof}[Proof of Proposition~\ref{P1}]
To begin with, observe that for $x<1$ the correspondence 
$$\vartheta\mapsto \left(1- x\right)\vartheta $$ 
maps $\AdmMMV$ onto itself by the cone property of $\AdmMMV$. Since we have $1-2\utilMMV = ((1-\id)\vee0)^2$, the expected utility of these two strategies satisfies
$$1-2\E[\utilMMV(x+\left(1- x\right)\vartheta\sint R_T)] 
= (1- x)^{2}\left(1-2\E[\utilMMV(\vartheta\sint R_T)]\right).$$
On taking infima over $\vartheta\in\AdmMMV$, one obtains 
\begin{equation}\label{eq:uMMVx}
1-2\uMMV(x) = (1-x)^{2}\left(1-2\uMMV(0)\right), \qquad x<1,
\end{equation}
and, should an optimal strategy $\aMMV(0)$ exist, also
$$ \aMMV(x) =\left(1-x\right)^+\aMMV(0),\qquad x\in\R,$$
since for $x\geq1$, it is optimal not to trade.

The MMV utility formula \eqref{eq:MMV2} and \eqref{eq:uMMVx} now yield 
\begin{equation}\label{eq:vMMV0}
 2\vMMV(0) = 2\sup_{x\in\R}\, \{\uMMV(x)-x\} 
= \left\{\begin{array}{cl} \left(1-2\uMMV(0)\right)^{-1}-1,&\qquad  2\uMMV(0)<1,\\[1ex]
\infty,&\qquad  2\uMMV(0)=1,
\end{array}\right.
\end{equation}
which shows the equivalence of \ref{T1.i} and \ref{T1.ii}.

For $2\uMMV(0)<1$, the supremum in \eqref{eq:vMMV0} is attained uniquely at  
$$\xMMV= 1-\left(1-2\uMMV(0)\right)^{-1}<1.$$
Thus if \ref{T1.i} or \ref{T1.ii} holds and $\aMMV(0)$ exists, then the optimal zero-cost risky investment strategy for the MMV utility is given by 
$$\bMMV(0) = \aMMV(\xMMV)=(1-\xMMV)\aMMV(0)=\left(1-2\uMMV(0)\right)^{-1}\aMMV(0)=\left(1+2\vMMV(0)\right)\aMMV(0).$$ 
The translation invariance of MMV utility now yields \eqref{eq:vMMVx} and \eqref{eq:bMMV}.

Conversely, suppose \ref{T1.i} or \ref{T1.ii} holds and a MMV optimizer $\bMMV(0)\in\AdmMMV$ exists. Then from \eqref{eq:MMV2}, there is $\hat \eta$ such that $\vMMV(0) = \E[\utilMMV(\bMMV(0)\sint R_T-\hat\eta)]+\hat\eta$, hence we may take $\aMMV(-\hat\eta)=\bMMV(0)$. Optimality (i.e., $\vMMV(0)\geq0$) and the upper bound $\utilMMV\leq \frac{1}{2}$ gives $\hat\eta\geq -\frac{1}{2}$. The first part of the proof now yields that we may take $\aMMV(0) = \frac{\aMMV(-\hat\eta)}{1+\hat\eta} = \frac{\bMMV(0)}{1+\hat\eta}$. 
\end{proof}

%=================%
%=================%

\section{Parametrization of trading strategies}\label{appx:B}

\begin{proof}[Proof of Proposition~\ref{P:lambda to alpha and back}]
Observe that $\utilMMV$ and $\utilMMV'-1= -(\id\wedge1)$ are universal representing functions in the sense of \cite[Definition~3.4]{cerny.ruf.22.ejp}. 
We have that $\alpha$ integrates $R$ because $\lambda$ integrates $R$ and $\Exp\left( -\left(\id\wedge1\right) \circ (\lambda\sint R) \right)_{-}$ is locally bounded. Next, we have
\begin{equation*}
 \alpha  = \lambda\Exp\left( -\left(\id\wedge1\right) \circ (\lambda\sint R) \right)_{-} = \lambda\Exp\left( \left(-1\vee \id\right) \circ (-\lambda\sint R) \right) _{-},
\end{equation*}
by composition (Proposition~\ref{P:composition}).  Observe that $\utilMMV' = (1-\id)^+ = 1-\id\wedge1$. Lemma~\ref{L:SDE} with $X=-\lambda\sint R$ now gives 
\begin{equation}\label{eq:Grana240513}
\begin{split}
\utilMMV'(\alpha\sint R) ={}& (1-\alpha\sint R)^+ 
= (1-\lambda\Exp\left( \left(-1\vee \id\right) \circ (-\lambda\sint R) \right) _{-}\sint R)^+
\\
={}& \Exp\left( \left(-1\vee \id\right) \circ (-\lambda\sint R) \right) 
=\Exp\left( -\left(\id\wedge1\right) \circ (\lambda\sint R) \right),
\end{split}
\end{equation}
which yields \eqref{eq:Granada241014}, \ref{20241230.a}, and \ref{20241230.b}. 
Since $2\utilMMV =1-(1-\id\wedge1)^2=1-(\utilMMV')^2$, the utility of wealth satisfies
\begin{equation*}
\begin{split}
2\utilMMV(\alpha\sint R) 
={}& 1- \Exp\left( -\left(\id\wedge1\right) \circ (\lambda\sint R) \right) ^2\\
={}&1-\Exp\big( \big(\left(1-\id\wedge1\right)^{2}-1\big) \circ (\lambda\sint R) \big),
\end{split}
\end{equation*}
where the first equality follows from \eqref{eq:Grana240513} and the second  from \cite[Proposition~4.3 and Corollary~3.20]{cerny.ruf.22.ejp}. The explicit form of $\utilMMV$ now yields \eqref{eq:terminal utility}.

To show the converse statement, \ref{20241230.a} yields that $\lambda$ is well defined and integrates $R$. By associativity of stochastic integrals, on $\lc0,\sigma\rc$ one has 
$$1-\alpha\sint R = 1+\left(1-\alpha\sint R_-\right)\sint (-\lambda\sint R) = \Exp\left(-\lambda\sint R\right).$$
This and \eqref{eq:Cierne250101} gives 
\[ \alpha\indicator{\lc0,\sigma\rc} = \lambda (1-\alpha\sint R_-)\indicator{\lc0,\sigma\rc} = \lambda\Exp\left(-\lambda\sint R\right)_-\indicator{\lc0,\sigma\rc} = \lambda \Exp\left( -\left(\id\wedge1\right) \circ (\lambda\sint R) \right)_{-}\] 
and \ref{20241230.b} yields \eqref{eq:alpha}.
\end{proof}

%=================%
%=================%

\section{Local expected utility maximization}\label{appx:C}
\begin{proof}[Proof of Proposition~\ref{P:Cierne240807}]
For deterministic $\lambda$, the realized local utility has independent increments by \cite[Proposition~2.15]{cerny.ruf.23.spa}. Observe further that $2\utilMMV(\lambda\Delta R_t)\leq \max_{x\in\R}2\utilMMV(x)=1$ for all $t\in[0,T]$ and that $\P[2\utilMMV(\lambda\Delta R_t)= 1]=1$ only if there is instantaneous arbitrage at time $t$, i.e., $\P[\lambda\Delta R_t\geq 1] = 1$. In view of \eqref{eq:NA_t}, we therefore have $\P[2\utilMMV(\lambda\Delta R_t)= 1]<1$ and hence $\E\left[-2\utilMMV(\lambda\Delta R_t)\right]>-1$, for all $t\in [0,T]$. Identity \eqref{eq:terminal utility} and  \cite[Theorem~4.1]{cerny.ruf.23.spa} with $Y=-2\utilMMV(\lambda\,\id) \circ R$ now yield the equivalence of \ref{P:Cierne240807.i} and \ref{P:Cierne240807.ii} as well as \eqref{eq:Cierne240808}.
\end{proof}

\begin{proof}[Proof of Proposition~\ref{P:Cierne240808}]
\noindent Step 1: We claim that the local expected utility  $\locutilMMV$ given by
\begin{equation*}
\locutil(\lambda)=  b^{ R[1]}\lambda - \frac{1}{2}\lambda c^{ R}\lambda^\top
+\int_{\R^d} \left(\utilMMV(\lambda x)-\lambda h(x)\right)F^{ R}(\d x),\qquad \lambda\in\R^d,
\end{equation*}
is a ``normal integrand'' on $[0,T]\times \R^d$ in the sense of \citet[Definition~14.27]{rockafellar.wets.98}. Indeed, 
\[
(t,\lambda)\mapsto \left(b_t^{ R[1]}+\int_{\R^d}(x\indicator{|x|\leq1}-h(x))F_t^X(\d x)\right)\lambda - \frac{1}{2}\lambda c_t^{ R}\lambda^\top
\] 
is a normal integrand by \cite[Example~14.29]{rockafellar.wets.98}. Furthermore, we have 
\[ \utilMMV(\lambda x)-\lambda x\indicator{|x|\leq1}\leq \frac{1}{2}\indicator{|x|>1}\leq |x|^2\wedge1,\qquad  (x,\lambda)\in\R^d\times\R^d,\]
$t\mapsto (|\id|^2\wedge1) F_t^R$ is a finite transition kernel from $([0,T],\mathcal{B}([0,T]))$ to $(\R^d,\mathcal{B}(\R^d))$, and   
\[(x,\lambda)\mapsto\frac{\utilMMV(\lambda x)-\lambda x\indicator{|x|\leq1}}{|x|^2\wedge1}\indicator{x\neq0},\] 
being measurable in $x$ and continuous in $\lambda$, is a normal integrand on $\R^d \times \R^d$ by \cite[Example~14.29]{rockafellar.wets.98}. By \citet[Lemma~2.55 and Theorem~1.20.4]{pennanen.perkkio.24}, 
\[
(t,\lambda)\mapsto \int_{\R^d} \left(\utilMMV(\lambda x)-\lambda h(x)\right)F_t^{ R}(\d x)
=\int_{\R^d} \frac{\utilMMV(\lambda x)-\lambda h(x)}{|x|^2\wedge1}\indicator{x\neq0}(|x|^2\wedge1)F_t^{ R}(\d x) 
\]
is a normal integrand on $[0,T]\times \R^d$, which establishes the claim.\medskip

\noindent Step 2: We claim that $\lambda\mapsto \locutilMMV_t(\lambda)$ has a maximizer for every $t\in[0,T]$.
Fix $\lambda\in\R^d$. To simplify the notation, observe that by \cite[Corollary~5.8]{cerny.ruf.22.ejp} one has
\begin{align*}
b^{(\lambda R)[1]}&{} =   b^{ R[1]}\lambda +\int_{\R^d} \left(\lambda x\indicator{|\lambda x|\leq 1}-\lambda h(x)\right)F^{ R}(\d x),\\
c^{\lambda R}   &{}   =  \lambda c^R\lambda^\top,\\
F^{\lambda R} &{} \text{ is the image measure of $F^R$ via the mapping $\lambda\mkern2mu \id$}.
\end{align*}
This yields
\[ \locutil(\lambda y) = b^{(\lambda R)[1]}y-\frac{1}{2}c^{\lambda R}y^2 + \int_{\R} (\utilMMV(xy)-xy\indicator{|x|\leq1})F^{\lambda R}(\d x), \quad y\in\R.\]

We shall now examine the behaviour of $\psi(y) = \locutil(\lambda y)$ for $y>0$ in a number of mutually exclusive subcases.
\begin{enumerate}[(a)]
\item For $c^{\lambda R}>0$ or $F^{\lambda R}(\R_-)>0$, Lemma~\ref{L:230717} yields $\lim_{y\to\infty}\frac{\psi(y)}{y}=-\infty$.
\item The case $c^{\lambda R}=0$ and $F^{\lambda R}(\R_-)=0$ admits two subcases as per condition \eqref{eq:instNA_PII}.  
\begin{enumerate}[(i)] 
\item Either $F^{\lambda R}(\R_+)>0$ and $b^{(\lambda R)[1]} - \int_{\R} x\indicator{|x|\leq 1}F^{\lambda R}(\d x)<0$, yielding $\lim_{y\to\infty}\frac{\psi(y)}{y} <0$ by Lemma~\ref{L:230717}; or
\item $F^{\lambda R}(\R_+)=0$ and $b^{(\lambda R)[1]}=0$, in which case $\psi(y)=0$ for all $y\in\R$. 
\end{enumerate}
\end{enumerate}
For each fixed $\lambda$ and all $t\in[0,T]$, we therefore have either $\lim_{y\to\infty} \frac{\locutil_t(\lambda y)}{y}<0$, which yields  $\lim_{y\to\infty} \locutil_t(\lambda y) = -\infty$, or $\locutil_t(\lambda y)=0$ for all $y\in\R$. By \citet*[Theorem~27.3]{rockafellar.70}, the function $\locutil_t$ attains its maximum in $\R^d$ for all $t\in[0,T]$.\medskip 

\noindent Step 3: By \citet[Theorem~14.37]{rockafellar.wets.98}, measurable selection applied to the normal integrand $\locutilMMV$ establishes \ref{Cierne240808.1}. 
By optimality, $0=\locutilMMV(0)\leq\locutilMMV(\lMMV)$. Since $\utilMMV$ is bounded from above, we have $\locutilMMV(\lMMV)<\infty$, yielding \ref{Cierne240808.2}.  
The first claim in \ref{Cierne240808.3} follows from $0\leq\locutilMMV(\lMMV)$ together with the inequality 
\[-\id^2\indicator{\id\leq-2}\leq \utilMMV\indicator{\id\leq-2}\leq-\frac{1}{2}\id^2\indicator{\id\leq-2}. \] 
Then $\psi$ given by $\psi(y) = \locutil(\lMMV y)$ for $y\geq 0$ is finite-valued and differentiable on $(0,\infty)$. By optimality of $\lMMV$, $\psi$ attains it maximum at $1$, hence $\psi'(1)=0$. Differentiation under the integral sign completes the proof.
\end{proof}

The next lemma studies the asymptotic slope of the local expected utility in a given direction. In this paper, it is used with $X=\lambda R$.

\begin{lemma}\label{L:230717}
Let $X$ be an $\R$-valued independent-increments semimartingale with characteristics $(b^{X[1]},c^X,F^X,A)$. Consider the function $\psi:\R\to [-\infty,\infty)$ given by 
\begin{equation*} 
\psi(y)= b^{X[1]}y - \frac{1}{2}  c^X y^2 +\int_{\R} \left(\utilMMV(xy)-xy\indicator{|x|\leq 1}\right)F^X(\d x).
\end{equation*}
Then 
\begin{equation}\label{eq:230718}
\begin{split}
 \infty>\lim_{y\to\infty}\frac{\psi(y)}{y} ={}& b^{X[1]}-\int_{\R} x\indicator{0<x\leq 1}F^X(\d x)
-\infty\left(\indicator{c^X>0}+\indicator{F^X(\R_-)>0}\right),
\end{split}
\end{equation}
with the convention $\infty\times0  = 0$.  
\end{lemma}
\begin{proof} 
Observe that $0\geq y^{-1}\left(\utilMMV(xy)-xy\right)\indicator{0<x\leq 1}\downarrow -x\indicator{0<x\leq 1}$ for $y\uparrow \infty$,
and, similarly, 
$$0\geq y^{-1}\left(\utilMMV(xy)-xy\right)\indicator{-1\leq x<0}\downarrow -\infty\indicator{-1\leq x<0},\qquad  y\uparrow \infty.$$
Monotone convergence now yields 
\begin{equation}\label{eq:mono1}
\begin{split}
\lim_{y\to\infty} y^{-1}\int_{\R} \left(\utilMMV(xy)-xy\right)\indicator{|x|\leq 1}F^X(\d x) ={}& -\int_{\R} x\indicator{0<x\leq 1}F^X(\d x)
-\infty \indicator{F^X([-1,0))>0}.
\end{split}
\end{equation}
The right-hand side is well defined with the convention $\infty\times0  = 0$ in the last term.

One also has 
$ 0\leq y^{-1}\utilMMV(xy)\indicator{1<x}\leq \frac{1}{2}\indicator{1<x}$ for $y\geq 1$.
Since $F^X([1,\infty))<\infty$, dominated convergence now yields
\begin{equation*}
\lim_{y\uparrow\infty} y^{-1}\int_{\R} \utilMMV(xy)\indicator{1<x}F^X(\d x) = 0.
\end{equation*}

As the final ingredient, monotone convergence yields
\begin{equation}\label{eq:mono2}
0\geq \lim_{y\to\infty} y^{-1}\int_{\R} \utilMMV(xy)\indicator{x<-1}F^X(\d x) = -\infty \indicator{F^X((\infty,-1))>0},
\end{equation}
with the convention $\infty\times0  = 0$.

Pulling \eqref{eq:mono1}--\eqref{eq:mono2} together yields \eqref{eq:230718}.
\end{proof}

%=================%
%=================%

\section{\texorpdfstring{$\sigma$}{Sigma}-special processes and semimartingale characteristics}\label{appx:D}
With \citet*{js.03}, we write $(b^{X[1]},c^X,F^X,A)$ for the characteristic triplet of some $\R^d$-valued semimartingale $X$ relative to some activity process $A$. Recall $X[1]=h\circ X$ with $h$ given in \eqref{eq:truncation}. The process $X[1]$ is special and its predictable compensator is precisely $B^{X[1]}= b^{X[1]}\sint A$. Furthermore, for special $X$ one clearly has 
\begin{equation}\label{eq:bX}
b^X = b^{X[1]} + \int_{\R^d} (x-h(x))F^X(\d x).
\end{equation} 
\begin{definition}
We say that a semimartingale $X$ is $\sigma$-special if there is a sequence of predictable sets $\{D_n\}_{n\in\N}$ increasing to $\Omega\times [0,\infty)$ such that $\indicator{D_{n}}\sint X$ is special for all $n\in\N$. 
\end{definition}
Observe that $|\id-h|\leq |\id|\indicator{|\id|>1}\leq |\id-h|+d\indicator{|\id|>1}$ for the $L^1$ norm on $\R^d$. Therefore, by \citet*[Lemma~3.2]{kallsen.04}, $X$ is $\sigma$-special if and only if $\int_{\R^d} |x-h(x)|F^X(\d x)$ is finite $(\P\otimes A)$-almost everywhere and $X$ is special if and only if $\int_{\R^d} |x-h(x)|F^X(\d x)$ integrates $A$. This motivates the following proposition.

\begin{proposition}\label{P:sigma-special}
If a semimartingale $X$ is $\sigma$-special, then $b^X$ is well defined by \eqref{eq:bX}. Furthermore, for any predictable set $D$ such that $\indicator{D}\sint X$ is special, we have $b^X\indicator{D}\in L(A)$ and
$ B^{\indicator{D}\sint X} = (b^{X}\indicator{D})\sint A$.
\end{proposition}

\begin{proof}
If $X$ is $\sigma$-special, then the right-hand side of \eqref{eq:bX} is well defined by the argument above. Observe that
\begin{align*}
B^{\indicator{D}\sint X} ={}& b^{\indicator{D}\sint X}\sint A = \left( b^{(\indicator{D}\sint X)[1]} + \int_{\R^d} (x-h(x))F^{\indicator{D}\sint X}(\d x)\right)\sint A\\
={}&\left( b^{X[1]}\indicator{D} + \int_{\R^d} (x-h(x))\indicator{D}F^{X}(\d x)\right)\sint A,
\end{align*}
yielding the statement.
\end{proof}

Our aim is to provide a systematic way of computing compensators of $\sigma$-special semimartingales obtained from certain representations that one can loosely think of as variations. To this end, we first introduce a (sub)class of predictable functions that is sufficiently general for our purposes.
\begin{definition}
Let $\tI^{n}(X)$ denote the set of all predictable functions $\xi: [0,\infty)\times\Omega\times\R^d \rightarrow \R^n$  such that the following properties hold.
	\begin{enumerate}[label={\rm(\arabic{*})}, ref={\rm(\arabic{*})}]
		\item  $\xi(0) = 0$.
		\item $x \mapsto \xi(x)$ is twice differentiable at zero.
		\item $D \xi(0) \in L(X)$.
		\item $D^2 \xi(0) \in L([X,X]^c)$.
		\item $(\xi -   D  \xi(0) \, \id) \in L(\mu^{{X}})$.
	\end{enumerate}
We write $\tI(X)=\bigcup_{n\in\N}\tI^n(X)$ and for $\xi\in\tI(X)$ we define 
\begin{equation}\label{eq:circ tI}
\xi\circ X =  D\xi(0)\sint X + \frac{1}{2}D^2\xi(0)\sint [X,X]^c + (\xi - D\xi(0)\id)*\mu^X.
\end{equation}
\end{definition}

The following proposition provides sufficient conditions for a composition of functions in $\tI$ to remain in $\tI$. For a proof see \cite[Theorem~3.18]{cerny.ruf.22.ejp}.
\begin{proposition}\label{P:composition}
	Let $\xi \in \tI(X)$. Moreover, fix $\psi \in \tI(\xi\circ X)$ such that $D\psi(0)$ is locally bounded.
		Then $\psi(\xi) \in \tI(X)$ and we have $\psi \circ (\xi \circ X) = \psi(\xi) \circ X$.
\end{proposition}

\begin{remark}
In \cite{cerny.ruf.22.ejp}, the $\circ$ operation is defined for predictable functions from a wider class $\I(X)\supset \tI(X)$. Observe that the two definitions are consistent on $\tI(X)$. Indeed, by \cite[Definition~3.8]{cerny.ruf.22.ejp}, for $\R^d$-valued $X$ and $\R$-valued $\xi\in\I(X)$ one has
\begin{equation}\label{eq:FlightToBA241229}
\xi\circ X =  \indicator{\mathcal{H}_X}D\xi(0)\sint X + \frac{1}{2}D^2\xi(0)\sint [X,X]^c 
+ (\xi - \indicator{\mathcal{H}_X}D\xi(0)\id)\star \mu^X.
\end{equation}
where $\mathcal{H}_X=\{\nu^X(\{\cdot\})=0\}$ and $\star$ is the $\sigma$-localised form of the finite variation jump integral $*$. We further have $\indicator{\mathcal{H}_X^c}\id \in L_\sigma(\mu^X)$ by \cite[Proposition~3.8]{cerny.ruf.21.bej}. From this and $\indicator{\mathcal{H}_X^c}D\xi(0)\in L(X)$, we obtain $\indicator{\mathcal{H}_X^c}D\xi(0)\id\in L_\sigma(\mu^X)$ and $\indicator{\mathcal{H}_X^c}D\xi(0)\sint X = \indicator{\mathcal{H}_X^c}D\xi(0)\id\star\mu^X$ by \cite[Proposition~3.9]{cerny.ruf.21.bej}. The consistency of the two definitions now follows from \eqref{eq:FlightToBA241229}.
\end{remark}

\begin{lemma}\label{L:240130}
For an $\R^d$-valued semimartingale $X$ and a predictable $\R$-valued function $\xi\in\tI(X)$, the following are equivalent.
\begin{enumerate}[(i)]
\item\label{L:240130.i} $\xi\circ X$ is $\sigma$-special;
\item\label{L:240130.ii} $\int_{\R^d} |\xi(x)|\indicator{|\xi(x)|>1}F^X(\d x)$ is finite $(\P\otimes A)$-almost everywhere.
\end{enumerate}

Furthermore, if either of the two conditions holds, then  
\begin{equation}\label{eq:Cierne251205}
b^{\xi\circ X} = D\xi(0)b^{X[1]}+\frac{1}{2}\mathrm{trace} \big(D^2\xi(0)c^{X}\big) + \int_{\R^d} \left(\xi(x)-D\xi(0)h(x)\right)F^X(\d x)
\end{equation}
and $\xi\circ X$ is special if and only if $\int_{\R^d} |\xi(x)|\indicator{|\xi(x)|>1}F^X(\d x)$ integrates $A$.
\end{lemma} 

\begin{proof}
Observe that $F^{\xi\circ X}$ is the push-forward measure of $F^X$ via the mapping $\xi$ (\cite[Corollary~5.8]{cerny.ruf.22.ejp}), hence $\int_{\R} |y|\indicator{|y|>1}F^{\xi\circ X}(\d y) = \int_{\R^d} |\xi(x)|\indicator{|\xi(x)|>1}F^X(\d x)$. The equivalence between \ref{L:240130.i} and \ref{L:240130.ii} as well as the  equivalence in the last claim now follow from \citet*[Lemma~3.2]{kallsen.04}. 
To argue the formula for $b^{\xi\circ X}$, assume that $\xi\circ X$ is $\sigma$-special. By Proposition~\ref{P:sigma-special}, $b^{\xi\circ X} = b^{(\xi\circ X)[1]}+\int_{\R} y\indicator{|y|>1}F^{\xi\circ X}(\d y)$ is well defined. Observe that $X-X[1]$ has finitely many jumps on compacts, hence $D\xi(0)\in L(X[1])$. The \'Emery formula \eqref{eq:circ tI} now yields
\[
(\xi\circ X)[1] =  D\xi(0)\sint X[1] + \frac{1}{2}D^2\xi(0)\sint [X,X]^c + (\xi\indicator{|\xi|\leq1} - D\xi(0)h) * \mu^X,
\] 
whereby one obtains by $\sigma$-localisation that
$$ b^{(\xi\circ X)[1]} = D\xi(0)b^{X[1]}+\frac{1}{2}\mathrm{trace} \big(D^2\xi(0)c^{X}\big) 
+ \int_{\R^d} \left(\xi(x)\indicator{|\xi(x)|\leq1}-D\xi(0)h(x)\right)F^X(\d x).$$
By image measure properties, $\int_{\R} y\indicator{|y|>1}F^{\xi\circ X}(\d y) = \int_{\R^d} \xi(x)\indicator{|\xi(x)|>1}F^X(\d x)$, which completes the proof of the drift formula. 
\end{proof}

We next formulate drift computation under an absolutely continuous change of measure, which will be used in the proof of Theorem~\ref{T6}.

\begin{lemma}\label{L:Metabief250115}
Let $Y$ be special such that $\Delta Y\geq -1$, $\Delta B^Y>-1$ and $M=\Exp(Y)/\Exp(B^Y)$ is a uniformly integrable martingale. Define $\Q$ by $\d\Q/\d\P=M_\infty$. For a $\P$-semimartingale $U$, the following are equivalent. 
\begin{enumerate}[(i)]
\item\label{L:Metabief250115:i} $U$ is $\Q$-$\sigma$-special.
\item\label{L:Metabief250115:ii} $U +  [U , Y]$ is $\P$-$\sigma$-special on the stochastic interval $M_->0$.
\end{enumerate}
If one of these conditions holds, then one has
\begin{equation*}
b^{U,\Q} = \frac{b^{U +  [U , Y]}}{1+ \Delta B^{Y}}  \qquad \text{on the stochastic interval $M_->0$}.
\end{equation*}
\end{lemma}
\begin{proof}
Fix a predictable set $D$ and apply \cite[Theorem~5.4]{cerny.ruf.23.spa} with $Z = \Exp(Y)$ and $V=\indicator{D}\sint U$, observing that $\Log(Z)=Y$ on the stochastic interval $M_->0$. The identity $\indicator{D}\sint(U+[U,Y])=V+[V,Y]$ yields the equivalence between \ref{L:Metabief250115:i} and \ref{L:Metabief250115:ii}. Suppose now \ref{L:Metabief250115:i} or \ref{L:Metabief250115:ii} holds and let $D$ belong to the corresponding $\sigma$-localising sequence. \cite[Theorem~5.4]{cerny.ruf.23.spa} then yields 
\[ 
\indicator{D}b^{U,\Q} = b^{V,\Q} = \frac{b^{V +  [V , Y]}}{1+ \Delta B^{Y}} = \frac{\indicator{D}b^{U +  [U , Y]}}{1+ \Delta B^{Y}}
\qquad \text{on the stochastic interval $M_->0$,}
\]
which completes the proof.
\end{proof}

%=================%
%=================%

\section{Optimal local expected utility strategy integrates \texorpdfstring{$R$}{R}}\label{appx:E}

\begin{proof}[Proof of Proposition~\ref{P:London241109}]
Fix $n\in\N$. Item~\ref{Cierne240808.3} of Proposition~\ref{P:Cierne240808} and Lemma~\ref{L:240130} yield that 
$(\id\wedge1)\circ (\lMMVn\sint R)$ and $(\id\wedge1)^2\circ (\lMMVn\sint R)$ are $\sigma$-special with 
$b^{(\id\wedge1)\circ (\lMMVn\sint R)} = b^{(\id\wedge1)^2\circ (\lMMVn\sint R)}$, which in view of \eqref{eq:utilMMV} gives 
\begin{equation}\label{eq:241007.2}
2b^{\utilMMV\circ (\lMMVn\sint R)} = b^{(\id\wedge1)\circ (\lMMVn\sint R)} 
=  b^{(\id\wedge1)^2\circ (\lMMVn\sint R)}
\end{equation}
and also
\begin{equation}\label{eq:Cierne250104}
b^{\id\indicator{\id\leq1}\circ (\lMMVn\sint R)} = b^{(\id^2\indicator{\id\leq1})\circ (\lMMVn\sint R)}\geq0.
\end{equation}
The latter yields the first claim. 

The easy inequalities
$0\leq -\id \indicator{\id<-1}\leq \id^2 \indicator{\id<-1}\leq (\id\wedge1)^2$
give 
$$
0\leq -b^{\id\indicator{\id< -1}\circ (\lMMVn\sint R)}  \leq  b^{(\id\wedge1)^2\circ (\lMMVn\sint R)},
$$
while \eqref{eq:Cierne250104} can be rephrased as 
\[
0\leq b^{\id\indicator{|\id|\leq1}\circ (\lMMVn\sint R)}
= b^{\indicator{\id\leq 1}\id^2\circ (\lMMVn\sint R)}-b^{\id\indicator{\id< -1}\circ (\lMMVn\sint R)} .
\]
Combining the two, one has 
$$ 
0\leq b^{\id\indicator{|\id|\leq1}\circ (\lMMVn\sint R)}\leq 2 b^{(\id\wedge1)^2\circ (\lMMVn\sint R)}.
$$
In view of $b^{(\id^2\wedge1)\circ (\lMMVn\sint R)}\leq b^{(\id\wedge1)^2\circ (\lMMVn\sint R)}$ and \eqref{eq:241007.2}, the proof of the second inequality in \eqref{eq:key inequality} is complete.
On the other hand, we have
\[
2b^{\utilMMV\circ (\lMMVn\sint R)} ={} b^{(\id\wedge1)\circ (\lMMVn\sint R)}
\leq {} b^{\id\indicator{|\id|\leq1}\circ (\lMMVn\sint R)} + b^{\indicator{\id>1}\circ (\lMMVn\sint R)}
\]
and the first inequality in \eqref{eq:key inequality} follows easily.
\end{proof}

\begin{proof}[Proof of Theorem~\ref{T1}]
Lemmata~\ref{L:armistice1} and \ref{L:armistice2} with Proposition~\ref{P:London241109} yield the equivalence of \ref{T2.i} and \ref{T2.ii}. 
Item~\ref{Cierne240808.3} of Proposition~\ref{P:Cierne240808} and Lemma~\ref{L:240130} yield that 
$(\id\wedge1)\circ (\lMMV\sint R)$ and $(\id\wedge1)^2\circ (\lMMV\sint R)$ are $\sigma$-special with 
$b^{(\id\wedge1)\circ (\lMMV\sint R)} = b^{(\id\wedge1)^2\circ (\lMMV\sint R)}$,
which in view of \eqref{eq:utilMMV} gives
\begin{equation}\label{eq:Cierne250104b}
 2\locutilMMV(\lMMV) = 2b^{\utilMMV\circ (\lMMV\sint R)} = b^{(\id\wedge1)\circ (\lMMV\sint R)} =  b^{(\id\wedge1)^2\circ (\lMMV\sint R)}.
\end{equation}
Since
$
|\utilMMV|\indicator{|\utilMMV|> 1} \leq 2\mkern2mu\id^2\indicator{\id\leq0}\leq 2 (\id\wedge1)^2,
$
and 
$b^{(\id\wedge1)^2\circ (\lMMV\sint R)}= 2b^{\utilMMV\circ (\lMMV\sint R)}$, 
 one has that 
\[
b^{|\utilMMV|\indicator{|\utilMMV|>1}\circ (\lMMV\sint R)}\leq 2 b^{(\id\wedge1)^2\circ (\lMMV\sint R)}
= 4b^{\utilMMV\circ (\lMMV\sint R)}=4\locutilMMV(\lMMV)
\] 
integrates $A$. Lemma~\ref{L:240130} now yields that $\utilMMV\circ (\lMMV\sint R)$ is special. The proof that $(\id\wedge1)\circ (\lMMV\sint R)$ and  $(\id\wedge1)^2\circ (\lMMV\sint R)$ are special is analogous. The last assertion of the theorem now follows from the identity \eqref{eq:Cierne250104b}.
\end{proof}

The following proposition details one step in the proof of Theorem~\ref{T4}, whose notation we adopt without further reference.
\begin{proposition}\label{P:Cierne250106}
For the optimal local expected quadratic utility strategy $\lMV$ let $\lMVn = \lMV\indicator{|\lMV|\leq n}$ for $n\in\N$. Then we have $b^{\id\mkern1mu\indicator{|\id|\leq1}\circ (\lMVn\sint R)}\geq0$ and   
\begin{equation}\label{eq:Cierne250106}
b^{\utilMV\circ (\lMVn\sint R)}\leq b^{(\id^2\wedge1)\circ (\lMVn\sint R)} + b^{\id\mkern1mu\indicator{|\id|\leq1}\circ (\lMVn\sint R)} 
\leq 6 b^{\utilMV\circ (\lMVn\sint R)},\qquad n\in\N.
\end{equation}
\end{proposition}

\begin{proof}
Fix $n\in\N$ and let $X=\lMVn\sint R$. The first-order conditions yield that 
$\id\circ X$ and $\id^2\circ X$
are $\sigma$-special with $b^{\id\circ X} = b^{\id^2\circ X}$,
which in view of $\utilMV = \id-\frac{1}{2}\id^2$ gives 
$2b^{\utilMV\circ X} = b^{\id\circ X}  =  b^{\id^2\circ X}$
and also $b^{\id\indicator{\id\leq1}\circ X} \geq b^{(\id^2\indicator{\id\leq1})\circ X}\geq0$.
The latter yields the first claim. 

The easy inequalities $0\leq -\id \indicator{\id<-1}\leq \id^2 \indicator{\id<-1}$ give 
\[b^{\id\indicator{|\id|\leq1}\circ X} \leq b^{\id\circ X} + b^{\id^2 \indicator{\id<-1}\circ X}\leq4b^{\utilMV\circ X}.\]
Together with $b^{(\id^2\wedge1)\circ X}\leq b^{\id^2\circ X} = 2 b^{\utilMV\circ X}$, this establishes the second inequality in 
\eqref{eq:Cierne250106}. On the other hand, $\utilMV\leq\id\wedge1$ yields
$ 
b^{\utilMV\circ X} \leq{} b^{(\id\wedge1)\circ X}
\leq {} b^{\id\indicator{|\id|\leq1}\circ X} + b^{\indicator{\id>1}\circ X}
$
and the first inequality in \eqref{eq:Cierne250106} follows easily.
\end{proof}

%=================%
%=================%

\section{Optimal monotone mean--variance portfolios}\label{appx:F}
\begin{proof}[Proof of Theorem~\ref{T2}]
Assume $\int_0^T \locutilMMV_t(\lMMV_t) \d A_t< \infty$ and let 
\begin{equation}\label{eq:a(0)}
\aMMV(0) = \lMMV\Exp\bigs( -\left(\id\wedge1\right) \circ (\lMMV\sint R) \bigs)_{-}.
\end{equation}
Then by Theorem~\ref{T1} and Proposition~\ref{P:lambda to alpha and back},  $\aMMV(0)$ integrates $R$ and $\utilMMV\circ(\lMMV\sint R)$ is special. By Proposition~\ref{P:Cierne240807}, one has $\utilMMV(\aMMV(0)\sint R_T)\in L^1$ with  
\begin{equation}\label{eq:uMMV(0)}
\E[2\utilMMV(\aMMV(0)\sint R_T)] = 1-\Exp\bigs(-2B^{\utilMMV \circ  (\lMMV\sint R)}\bigs)_T<1.
\end{equation}

We next show admissibility of $\aMMV(0)$. Since $\utilMMV(\aMMV(0)\sint R_T)\in L^1$ and $(\utilMMV')^2=1-2\utilMMV$, Proposition~\ref{P:lambda to alpha and back} gives 
\[
\utilMMV'(\aMMV(0)\sint R_T)=\Exp(U)_T\in L^2
\]
with $U=-\left(\id\wedge1\right) \circ (\lMMV\sint R) = (\utilMMV'-1)\circ (\lMMV\sint R)$. We claim $\Delta B^U>-1$.
Indeed, the first-order conditions for $\lMMV$ in \eqref{eq:Granada241018.2} give
\[-\Delta B^U_t = \E[-\Delta U_t] = \E[(\Delta U_t)^2], \]
and $\Var(\Delta U_t)\geq0$ yields  $\Delta B^U_t\geq -1$, with $\Delta B^U_t = -1$ if and only if $\Delta U_t = -1$. The latter is ruled out by the absence of instantaneous arbitrage, which proves $\Delta B^U>-1$. \citet*[Theorem~4.1]{cerny.ruf.23.spa} now yields that $U$ is special, hence $\Exp(B^U)>0$. The multiplicatively compensated process $\frac{\Exp(U)}{\Exp(B^U)}$ is then a (square-integrable) martingale by \cite[Theorems~1.2~and~4.1]{cerny.ruf.23.spa}. In view of $\Exp(B^U)>0$, the Doob inequality yields 
\[
L^2\ni\sup_{t\in[0,T]}|\Exp(U)_t| = \sup_{t\in[0,T]}|\utilMMV'(\aMMV(0)\sint R)_t|,
\] 
hence $\aMMV(0)\in\AdmMMV$ by \eqref{eq:250111.MMV}.

As the next step, we shall show that $\aMMV(0)$ has the highest expected utility in $\AdmMMV$, i.e.,
\begin{equation}\label{eq:max utility}
\max_{\alpha\in\AdmMMV}\E[\utilMMV(\alpha\sint R_T)] = \E[\utilMMV(\aMMV(0)\sint R_T)].
\end{equation}
Fix an arbitrary $\alpha\in\AdmMMV$ and denote by  $\tau$  the first time $\utilMMV'(\alpha\sint R)=(1-\alpha\sint R)^+$ becomes zero. Observe that $\alpha\indicator{\lc 0, \tau\rc}\in\AdmMMV$ and that the terminal utility of $\alpha\indicator{\lc 0, \tau\rc}$ is at least as high as that of $\alpha$. Furthermore,
\[
Z= \left(\utilMMV'(\alpha\indicator{\lc 0, \tau\rc}\sint R)\right)^2 = 1-2\utilMMV(\alpha\indicator{\lc 0, \tau\rc}\sint R)
\]
is absorbed in zero if it ever hits zero. Admissibility of $\alpha\indicator{\lc 0, \tau\rc}$ means that $\sup_{t\in[0,T]}Z_t$ is integrable, hence $Z\geq0$ is special. 

We shall now apply the converse direction of Proposition~\ref{P:lambda to alpha and back} on a certain stochastic interval.
Denoting by $\tau^Z_c$ the first time $Z$ reaches zero continuously, we have that $\Log(Z)$ is well defined and special on $\lc 0,\tau^Z_c\lc$. Letting
\[
\lambda = \frac{\alpha\indicator{\lc 0, \tau\rc}}{1-\alpha\indicator{\lc 0, \tau\rc}\sint R_-}\indicator{\lc0,\tau^Z_c\lc},
\]
we obtain  
\[ 
B^{\Log(Z)}\indicator{\lc0,\tau^Z_c\lc} = -2B^{\utilMMV \circ (\lambda\sint R)}\indicator{\lc0,\tau^Z_c\lc}
\]
by Proposition~\ref{P:lambda to alpha and back}. By optimality of $\lMMV$ we have 
\begin{equation}\label{eq:Granada241012}
(2B^{\utilMMV \circ (\lMMV\sint R)}-2B^{\utilMMV \circ (\lambda\sint R)})\indicator{\lc0,\tau^Z_c\lc} \text{ is non-decreasing on $\lc0,\tau^Z_c\lc$,} 
\end{equation}
hence by Lemma~\ref{L:London241228} applied on $\lc0,\tau^Z_c\lc$ with $X=-2B^{\utilMMV \circ (\lambda\sint R)}$ and $X'=-2B^{\utilMMV \circ (\lMMV\sint R)}$ one obtains
\begin{align*}
\liminf_{t\uparrow\tau^Z_c}\Exp(B^{\Log(Z)})_t={}& \liminf_{t\uparrow\tau^Z_c}\Exp(-2B^{\utilMMV \circ (\lambda\sint R)})_t\\
\geq {}&\liminf_{t\uparrow\tau^Z_c}\Exp(-2B^{\utilMMV \circ (\lMMV\sint R)})_t
\geq \Exp(-2B^{\utilMMV \circ (\lMMV\sint R)})_T>0.
\end{align*} 
 
\cite[Theorem~3.1]{cerny.ruf.23.spa} now yields that 
\[ M = \frac{Z}{\Exp(B^{\Log(Z)})}\indicator{\lc0,\tau^Z_c\lc}
\] is a local martingale on $[0,T]$. Furthermore, $M$ is uniformly integrable (hence a martingale) since $\sup_{t\in[0,T]}Z_t$ is integrable and $\frac{1}{\Exp(B^{\Log(Z)})}\indicator{\lc0,\tau^Z_c\lc}\leq \frac{1}{\Exp(-2B^{\utilMMV \circ (\lMMV\sint R)})_T}$ by \eqref{eq:Granada241012}.

Recalling \eqref{eq:uMMV(0)}, we have
\begin{align*}
\frac{1-2\E\left[\utilMMV(\alpha\indicator{\lc 0, \tau\rc}\sint R_T)\right]}{1-2\E\left[\utilMMV(\aMMV(0)\sint R_T)\right]} = 
\frac{\E\left[Z_T\right]}{\Exp(-2B^{\utilMMV \circ (\lMMV\sint R)})_T} = 
\E\left[M_T\frac{\Exp(B^{\Log(Z)})_T\indicator{\{\tau^Z_c>T\}}}{\Exp(-2B^{\utilMMV \circ (\lMMV\sint R)})_T}\right]\geq \E\left[M_T\right] = 1, 
\end{align*}
where the inequality follows from \eqref{eq:Granada241012}. After rearranging, one has 
\begin{equation*}
\E\left[\utilMMV(\alpha\sint R_T)\right] \leq \E\left[\utilMMV(\alpha\indicator{\lc 0, \tau\rc}\sint R_T)\right] \leq \E\left[\utilMMV(\aMMV(0)\sint R_T)\right]
\end{equation*}
for all $\alpha\in\AdmMMV$, which proves \eqref{eq:max utility}.
\smallskip

In \eqref{eq:uMMV(0)} and \eqref{eq:max utility} we have established $2\uMMV(0)=1-\Exp\bigs(-2B^{\utilMMV \circ  (\lMMV\sint R)}\bigs)_T<1$ and shown $\aMMV(0)\in\AdmMMV$ in \eqref{eq:a(0)} to be the maximizer. The rest of the statement now follows from Proposition~\ref{P1}.
\end{proof}

\begin{proof}[Proof of Theorem~\ref{T3}]
For each $n\in\N$, let 
\begin{equation*}
\lMMVn_t = \lMMV_t\indicator{|\lMMV_t|\leq n}\indicator{\locutilMMV_t(\lMMV_t)\leq n}.
\end{equation*}
Then $\lMMVn$ is bounded, therefore it integrates $R$. By Proposition~\ref{P:Cierne240808}, we have 
$$ 0\leq \locutilMMV_t(\lMMVn_t)\uparrow  \locutilMMV_t(\lMMV_t)<\infty,\qquad t\in[0,T], \quad n\uparrow \infty.$$
An analog of Theorem~\ref{T1} with $\lMMV$ replaced by $\lMMVn$ yields that $\utilMMV(\lMMVn\,\id)\circ R$ is special. 

Letting $\alpha^n = \lMMVn\Exp\bigs( -\left(\id\wedge1\right) \circ (\lMMVn\sint R) \bigs) _{-}$, Proposition~\ref{P:Cierne240807} and monotonicity of $B^{\utilMMV(\lMMVn\,\id) \circ R}$ yield 
$$2u^n=\E[2\utilMMV(\alpha^n\sint R_T)] = 1-\Exp\bigs( -2B^{\utilMMV(\lMMVn\,\id) \circ R} \bigs)_T\leq 
1-\exp\bigs( -2B^{\utilMMV(\lMMVn\,\id) \circ R}_T \bigs).$$
Since 
\[
B^{\utilMMV(\lMMVn\,\id) \circ R}_T = \int_0^T \locutilMMV_t(\lMMVn_t)\d A_t\uparrow \int_0^T \locutilMMV_t(\lMMV_t) \d A_t= \infty,
\] 
we have $2u^n\uparrow1$, which proves the first assertion in view of \eqref{eq:utilMMV}. 

To see the second assertion, observe that by an analogy of Proposition~\ref{P1} with $\bigcup_{\lambda\geq 0}\lambda \alpha^n$ in place of $\AdmMMV$, $\alpha^n$ in place of $\aMMV(0)$, and $u^n$ in place of $\uMMV(0)$, the strategies $\beta^n = (1-2u^n)^{-1}\alpha^n$ satisfy $2V_{\MMV}(\beta^n\sint R_T) =\left(1-2u^n\right)^{-1}-1 \uparrow\infty$ as $n\uparrow \infty$ since $2u^n\uparrow1$.
\end{proof}

%=================%
%=================%

\section{Duality}\label{appx:G}

\begin{proof}[Proof of Proposition~\ref{P:Granada241017}]
Let $\Q$ be a probability measure with square-integrable density. The function $h_\Q:\R\to[0,\infty]$, 
\[
h_\Q(y)= \E\left[\frac{1}{2}\left(y\frac{\d\Q}{\d\P}-1\right)^2\right] + \infty\indicator{y<0},
\]
is proper convex and lower semicontinuous. Consider now the function $\uMMV^\Q:\R\to\R_+$, given by  
$$\uMMV^\Q(x)=\sup_{{X\in L^1(\Q),\, \E^{\Q}[X]\leq0}}\{\E[\utilMMV(x + X)]\},\qquad x\in\R.$$
In economic terms,  $\uMMV^\Q(x)$ measures the maximal expected monotone quadratic utility under $\P$ in a statically complete market with pricing measure $\Q$ for an agent with  initial wealth $x$. By \citet*[Lemma~4.3]{biagini.cerny.11} we have 
$\uMMV^\Q(x) = \min_{y\geq0} \left\{xy+h_\Q(y)\right\}$ for all $x\in\R$.
The Fenchel-Moreau theorem (e.g., \cite[Theorem~12.2]{rockafellar.70}) and the closedness of $h_\Q$ now yield 
\begin{equation}\label{eq:Granada241018}
\sup_{x\in\R}\,\{\uMMV^\Q(x)-xy\}=h_\Q(y),\qquad y\in\R.
\end{equation} 

For a separating measure $\Q$ we thus have 
\begin{align*}
\vMMV(0)  =   {}& \sup_{\vartheta\in\AdmMMV}\,\{V_{\MMV}(\vartheta\sint R_T)\}
\leq  {} \sup_{{X\in L^1(\Q),\, \E^{\Q}[X]\leq0}}\, \{V_{\MMV}(X)\}\\
					=   {}& \sup_{x\in\R}\,\{\uMMV^\Q(x)-x\} = h_{\Q}(1) = \frac{1}{2}\Var\left(\frac{\d\Q}{\d\P}\right),
\end{align*}
where the inequality follows by inclusion, the second equality follows from \eqref{eq:MMV2}, and the last two equalities follow from \eqref{eq:Granada241018} and the explicit form of $h_\Q$.
\end{proof} 

\begin{proof}[Proof of Theorem~\ref{T5}]
(1): By Theorems~\ref{T2} and \ref{T3}, one has $\vMMV(0)<\infty$ if and only if $\int_0^T \locutilMMV_t(\lMMV_t) \d A_t<\infty$. Letting 
$$U=-\left(\id\wedge1\right) \circ(\lMMV\sint R),$$ 
it is shown in the proof of Theorem~\ref{T2} that $U$  is special with $\Delta B^U>-1$ and that $\vfrac{\Exp(U)}{\Exp(B^U)}$ is a square-integrable martingale, which yields the claim.\smallskip

\noindent (2): By Theorem~\ref{T2}, $\aMMV(0) = \lMMV\Exp(U)_-$ maximizes the expected monotone quadratic utility. Proposition~\ref{P:lambda to alpha and back} further gives $\utilMMV'(\aMMV(0)\sint R_T) = \Exp(U)_T=\Exp(B^U)_T\vfrac{\d\QMMV}{\d \P}$. Fix $\vartheta\in\AdmMMV$. By optimality of $\aMMV(0)$, the function $\psi:[0,\infty) \to \R$ given by $\psi(y) =\E[\utilMMV( \aMMV\sint R_T+y\vartheta\sint R_T)]$ attains its maximum at $0$. By dominated convergence and optimality, one obtains 
\[
0\geq \psi_+'(0) = \E[\utilMMV'( \aMMV(0)\sint R_T)\vartheta\sint R_T]>-\infty,
\] 
which shows that $\vartheta\sint R_T\in L^1(\QMMV)$ and $\E^{\QMMV}[\vartheta\sint R_T]<0$. Since $\vartheta\in\AdmMMV$ was arbitrary, this proves that $\QMMV$ is a separating measure. \smallskip

\noindent (3): Theorem~\ref{T2} together with \eqref{eq:Granada241018.2} yields 
\[
1+2\vMMV(0) = \Exp\bigs(-2B^{\utilMMV \circ  (\lMMV\sint R)}\bigs)_T^{-1} = \Exp\bigs(B^{U}\bigs)_T^{-1}.
\]
On the other hand, the Yor formula yields
\[ 
\Exp(U)^2 = \Exp\bigs( \left((1-(\id\wedge1))^2-1\right) \circ(\lMMV\sint R) \bigs) = \Exp\bigs( -2\utilMMV \circ(\lMMV\sint R) \bigs),
\]
which together with \cite[Theorem~4.1]{cerny.ruf.23.spa} and \eqref{eq:Granada241018.2} gives
\[
\E\bigg[\Big(\frac{\d \QMMV}{\d \P}\Big)^2\bigg] = \frac{\Exp\bigs(-2B^{\utilMMV \circ  (\lMMV\sint R)}\bigs)_T}{\Exp\bigs(B^{U}\bigs)_T^2} = \Exp\bigs(B^{U}\bigs)_T^{-1}=1+2\vMMV(0).
\]
This completes the proof of item (3).\smallskip

\noindent(4): In view of (3), Proposition~\ref{P:Granada241017} yields that the density of $\QMMV$ has the smallest variance among all separating measures. Arguing by contradiction, suppose there is another separating measure, say $\Q'$ whose density has the same variance. Strict convexity of $\id^2$ now yields that  $\QMMV/2 +\Q'/2$ is another separating measure with strictly smaller variance, which contradicts the already established optimality of $\QMMV$.\smallskip

\noindent(5): By \citet*[Proposition~4.2]{cerny.ruf.22.ejp}, for $\Delta R>-1$ one has $\log(\Exp(R)) = \log(1+\id)\circ  R$. The claim follows from the explicit form of $\QMMV$ and \cite[Proposition~2.15 and Corollary~5.10]{cerny.ruf.23.spa}.
\end{proof}

\begin{proof}[Proof of Theorem~\ref{T6}]
Fix $i\in\{1,\ldots,d\}$.  Letting $V=R^{(i)}$ and $Y=-\left(\id\wedge1\right) \circ(\lMMV\sint R)$, we have $V+[V,Y] = \id_i\utilMMV'(\lMMV\mkern1mu\id)\circ R$.  Lemma~\ref{L:Metabief250115} now yields \ref{T6.ii} $\Rightarrow$ \ref{T6.i}. To argue the converse direction, Lemma~\ref{L:Metabief250115} yields that $b^{V+[V,Y]}$ is well-defined and equal to zero on all paths where $\d\QMMV/\d\P>0$, hence everywhere since $V+[V,Y]$ has independent increments under $\P$ and $\P[\d\QMMV/\d\P>0]>0$ by \cite[Lemma~4.2]{cerny.ruf.23.spa}. 
\end{proof}

\begin{proof}[Proof of Proposition~\ref{P:Metabief250115}]
Working under $\QMMV$, we have that $\vartheta\sint R$ has independent increments by Theorem~\ref{T5}\ref{T5.3} and $\vartheta\sint R_T$ is integrable by  Theorem~\ref{T5}\ref{T5.2}. By Theorem~\ref{T:250107}, the  $\sigma$-martingale $\vartheta\sint R$ is a martingale.
Next, fix $i\in\{1,\ldots,d\}$. If the jumps of $R^{(i)}$ are bounded from below, then  $R^{(i)}=1\sint R^{(i)}$ is the wealth of a deterministic admissible strategy by Lemma~\ref{L:Cierne250106}, and the claim follows. 
\end{proof}

\begin{proof}[Proof of Theorem~\ref{T7}]
Let $\Q$ be a (signed) measure with square-integrable density. Consider the function $\uMV^\Q:\R\to\R_+$, given by  
$$\uMV^\Q(x)=\sup_{{X\in L^2(\P),\, \E^{\Q}[X]=0}}\{\E[\utilMV(x + X)]\},\qquad x\in\R.$$
By the Fenchel inequality we have $\utilMV(x + X)\leq (x+X)y\frac{\d\Q}{\d\P} + \frac{1}{2}(1-y\frac{\d\Q}{\d\P})^2$ for all $y\in\R$ and $X \in L^2(\P)$, hence
$\uMV^\Q(x) \leq \inf_{y\in\R}\left\{xy +\frac{1}{2}\E\left[\left(1-y\frac{\d\Q}{\d\P}\right)^2\right]\right\},$
whereby the Fenchel--Moreau theorem (e.g., \cite[Theorem~12.2]{rockafellar.70}) gives  
\begin{equation}\label{eq:Granada251217}
\sup_{x\in\R}\,\{\uMV^\Q(x)-xy\}\leq \frac{1}{2}\E\left[\left(1-y\frac{\d\Q}{\d\P}\right)^2\right],\qquad y\in\R.
\end{equation}   
If $\Q$ is separating over $\AdmMV$, we thus have 
\begin{align*}
\vMV(0)  =   {}& \sup_{\vartheta\in\AdmMV}\,\{V_{\MV}(\vartheta\sint R_T)\}
\leq  {} \sup_{{X\in L^2(\P),\, \E^{\Q}[X]=0}}\, \{V_{\MV}(X)\}\\
					=   {}& \sup_{x\in\R}\,\{\uMV^\Q(x)-x\} \leq  \frac{1}{2}\Var\left(\frac{\d\Q}{\d\P}\right),
\end{align*}
where the first inequality follows by inclusion,  the second equality follows from \eqref{eq:MV2}, and the second inequality is obtained from \eqref{eq:Granada251217}.

Let now $U = (\utilMV'-1)\circ (\lMV\sint R) = -\lMV\sint R$. Repeating the argument in the second paragraph of the proof of Theorem~\ref{T2} yields $\Delta B^U>-1$. By Theorem~\ref{T4}, $\aMV(0) = \lMV\Exp(U)_-$ maximizes the expected quadratic utility. A calculation analogous to Proposition~\ref{P:lambda to alpha and back} gives 
$$\utilMV'(\aMV(0)\sint R_T) = \Exp(U)_T=\Exp(B^U)_T\frac{\d\QMV}{\d \P}.$$ 
 Optimality of $\aMV(0)$ yields $\E[\utilMV'(\aMV(0)\sint R_T)(\vartheta\sint R_T)]=0$ for all $\vartheta\in\AdmMV$, which shows that $\QMV$ is separating over $\AdmMV$. The equality $2\vMV(0) = \Var(\d\QMV/\d\P)$ follows by direct calculation involving expectations of stochastic exponentials of processes with independent increments. The last part of the proof is fully analogous to the proof of Theorem~\ref{T6}.
\end{proof}

\begin{proof}[Proof of Theorem~\ref{T8}]

\noindent \ref{T8.i} $\Rightarrow$ \ref{T8.ii} Observe that for $\sigma$-locally square-integrable $R$, the local utilities $\locutilMV$ and $\locutilMMV$ are finite and differentiable on the whole $\R^d$. \citet[Section~27, p.~264]{rockafellar.70} yields that $\partial\locutilMV(\lMV)=\partial\locutilMMV(\lMMV)=0$ characterizes all local optimizers. In view of $\utilMMV'-\utilMV'=(\id-1)\indicator{\id>1}$, the inequality $\lMV\Delta R\leq1$  gives 
\[0=\partial\locutilMV(\lMV)=b^{\id\utilMV'(\lMV\mkern1mu\id)\circ R}=b^{\id\utilMMV'(\lMV\mkern1mu\id)\circ R}=\partial\locutilMMV(\lMV),\] 
hence $\lMV$ is a local optimizer of $\locutilMMV$. \smallskip

\noindent \ref{T8.ii} $\Rightarrow$ \ref{T8.iii} This follows from the explicit expressions for $\vMV(0)$ and $\vMMV(0)$ in Theorems~\ref{T2} and \ref{T4}, where one may take $\lMMV=\lMV$ by \ref{T8.ii}, see also Remark~\ref{R:nonuniqueness}.\smallskip  

\noindent \ref{T8.iii} $\Rightarrow$ \ref{T8.iv} Theorems~\ref{T5} and \ref{T7} and \ref{T8.iii} yield $\Var(\frac{\d\QMV}{\d\P}) = \Var(\frac{\d\QMMV}{\d\P})$. Since $\AdmMV\subset\AdmMMV$, $\QMMV$ is a separating measure over $\AdmMV$. Strict convexity of the second moment and the minimality established in Theorem~\ref{T7} yield  $\QMMV=\QMV$. \smallskip

\noindent \ref{T8.iv} $\Rightarrow$ \ref{T8.v} Since $\QMMV$ is a probability measure, this is trivial.\smallskip

\noindent\ref{T8.v} $\Rightarrow$ \ref{T8.i} The martingale property of $\vfrac{\Exp( -\lMV\sint R )}{\Exp(B^{-\lMV\sint R})}$ and $\frac{\d\QMV}{\d\P}\geq0$ give $\Exp(-\lMV\sint R)\geq0$ as well as $\P[\Exp( -\lMV\sint R )_t>0]>0$ for all $t\in[0,T]$. This then yields $\lMV_t\Delta R_t\leq1$ on the set $\{\Exp( -\lMV\sint R )_t>0\}$, hence everywhere by independence of increments. \smallskip

\noindent The sequel follows by Theorem~\ref{T7} in view of the first-order condition $b^{\id\utilMV'(\lMV\mkern1mu\id)\circ R}=0$ argued in the proof of \ref{T8.i} $\Rightarrow$ \ref{T8.ii}.
\end{proof}

\section{Auxiliary statements}
\begin{lemma}\label{L:SDE}
For any semimartingale $X$, the SDE
\begin{equation}\label{SDE}
 W = 1+ W^+_-\sint X 
\end{equation}
has a unique solution 
\begin{equation}\label{eq:Granada240927}
W =  1 + \Exp((-1\vee\id)\circ X)_-\sint X
\end{equation}
that furthermore satisfies 
$$ W^+ = \Exp((-1\vee\id)\circ X).$$
\end{lemma}
\begin{proof}
Since the function $x\mapsto x^+$ is Lipschitz, the SDE \eqref{SDE} has a unique (semimartingale) solution $W$ by \citet*[Theorem~V.6]{protter.05}. Denote by $\tau$ the first time $\Delta X\leq -1$. Then  $\Exp(X^\tau)^+=\Exp(X^\tau)\indicator{\lc0,\tau\lc}$ by \citet*{js.03} and $\Exp(X^\tau)^+_-=\Exp(X^\tau)_-\indicator{\lc0,\tau\rc}$. Since 
$$\Exp(X^\tau) =  1+\Exp(X^\tau)_- \sint X^\tau = 1+\Exp(X^\tau)_-\indicator{\lc0,\tau\rc}\sint X,$$ 
we have $W=\Exp(X^\tau)$. Observe that $-1\vee\id$ is a universal representing function in the sense of \cite[Definition~3.4]{cerny.ruf.22.ejp}, hence 
$(-1\vee\id)\circ X$ is well defined. In view of \cite[Proposition~3.13(1)]{cerny.ruf.22.ejp}, $\tau$ is also the first time $(-1\vee\id)\circ X$ jumps by $-1$. By \citet*{js.03} again, $\Exp((-1\vee\id)\circ X)=\Exp(X)$ on $\lc0,\tau\lc$ and $\Exp((-1\vee\id)\circ X)_\tau=0$, hence 
\[ W^+ = \Exp(X^\tau)^+ =\Exp(X^\tau)\indicator{\lc0,\tau\lc} = \Exp((-1\vee\id)\circ X).\]
In combination with \eqref{SDE}, this yields \eqref{eq:Granada240927}. 
\end{proof}

We next characterize integrability in terms of predictable characteristics. We first present a somewhat simplified version of \citet*[Theorem~III.6.30]{js.03}, where we additionally interpret the predictable criterion as a certain predictable variation. 
\begin{theorem}\label{T:int}
For an  $\R^d$-valued predictable process $\zeta$ and  an $\R^d$-valued semimartingale $X$, the following are equivalent.
\begin{enumerate}[(i)]
\item\label{Tint.i} $\zeta\in L(X)$.
\item\label{Tint.ii} The predictable process $\eta$ defined by
\begin{align*}
\eta = \zeta c^X\zeta^\top + \int_{\R^d} ((\zeta x)^2\wedge1)F^X(\d x) 
+ \left| \zeta b^{X[1]} + \int_{\R^d} \left(\zeta x \indicator{|\zeta x|\leq 1} - \zeta h(x) \right)F^X(\d x)\right|
\end{align*} 
integrates $A$.
\end{enumerate}
Furthermore, if either of these conditions holds, then  one has 
\[
\eta = b^{(\id^2\wedge1)\circ (\zeta\sint X)} + \big|b^{\id\indicator{|\id|\leq1}\circ (\zeta\sint X)}\big|.
\]
\end{theorem}
\begin{proof}
\ref{Tint.i}$\implies$\ref{Tint.ii} and the sequel:\ $\zeta\in L(X)$ yields that $Y=\zeta\sint X$ is a semimartingale. Since $(\id^2\wedge1)\in\tI(\zeta\sint X)$ and $\id\indicator{|\id|\leq 1}\in\tI(\zeta\sint X)$, composition (Proposition~\ref{P:composition}) yields that $(\zeta\mkern2mu\id)^2\wedge1\in\tI(X)$, $\zeta\mkern2mu\id\indicator{|\zeta\mkern1mu\id|\leq1}\in\tI(X)$ and 
\begin{align*}
\left(\id^2\wedge1\right)\circ (\zeta\sint X) ={}& \left((\zeta\mkern2mu\id)^2\wedge1\right)\circ X,\\
\id\indicator{|\id|\leq1}\circ (\zeta\sint X) ={}& \zeta\mkern2mu\id\indicator{|\zeta\mkern1mu\id|\leq1}\circ X.
\end{align*}
These two semimartingales are special since their jumps are bounded. An explicit expression for the predictable compensator of these two quantities via Lemma~\ref{L:240130} yields \ref{Tint.ii} and the sequel. 

\ref{Tint.ii}$\implies$\ref{Tint.i} Let $\tilde h = \id\mkern2mu\indicator{|\id|\leq1}$ and $X[\tilde h] = \tilde h\circ X$. By invariance of truncations, one has 
\[
\eta = \zeta c^X\zeta^\top + \int_{\R^d} ( (\zeta x)^2\wedge1)F^X(\d x) 
+ \left| \zeta b^{X[\tilde h]} + \int_{\R^d} \left(\zeta x \indicator{|\zeta x|\leq 1} - \zeta \tilde h(x) \right)F^X(\d x)\right|.
\]
Observe that $\indicator{|\id|> 1}\in L(\nu^X)$, hence 
\[ \left|\int_{\R^d} \zeta x \indicator{|\zeta x|\leq 1}(\indicator{|x|\leq1}-1)F^X(\d x)\right| \leq \int_{\R^d} \indicator{|x|>1}F^X(\d x)\]
integrates $A$. Since $\eta$ integrates $A$, the triangle inequality yields
\[\zeta c^X\zeta^\top + \int_{\R^d} ( (\zeta x)^2\wedge1)\indicator{|x|\leq1}F^X(\d x) 
+ \left| \zeta b^{X[\tilde h]} + \int_{\R^d} \left(\zeta x \indicator{|\zeta x|\leq 1}\indicator{|x|\leq1} - \zeta \tilde h(x) \right)F^X(\d x)\right|  \]
integrates $A$. \citet*[Theorem~III.6.30]{js.03} now gives \ref{Tint.i}.
\end{proof}

\begin{corollary}\label{C:integrability}
Let $\zeta$ be an $\R^d$-valued predictable process and $X$ an $\R^d$-valued semimartingale. Then, for each $n\in\N$, $\zeta^n = \zeta\indicator{|\zeta|\leq n}$ integrates $X$ and the following are equivalent.
\begin{enumerate}[(i)]
\item $\zeta\in L(X)$.
\item $\lim_{n\uparrow \infty} \left(b^{(\id^2\wedge1)\circ (\zeta^n\sint X)} + \big|b^{\id\indicator{|\id|\leq1}\circ (\zeta^n\sint X)}\big|\right)$ integrates $A$.
\end{enumerate}
\end{corollary}
\begin{proof}
The boundedness of $\zeta^n$ for each $n\in\N$ yields $\zeta^n\in L(X)$. Theorem~\ref{T:int} with $\zeta^n$ in place of $\zeta$ now yields that 
\begin{align*}
\eta^n ={}& b^{(\id^2\wedge1)\circ (\zeta^n\sint X)} + \big|b^{\id\indicator{|\id|\leq1}\circ (\zeta^n\sint X)}\big|\\
={}&  \zeta^n c^X(\zeta^n)^\top + \int_{\R^d} ( (\zeta^n x)^2\wedge1)F^X(\d x) 
+ \left| \zeta^n b^{X[1]} + \int_{\R^d} \left(\zeta^n x \indicator{|\zeta^n x|\leq 1} - \zeta^n h(x) \right)F^X(\d x)\right|
\end{align*}
is well defined. In view of $0\leq\eta^n\uparrow\eta$, the claim follows by Theorem~\ref{T:int}.
\end{proof}

\begin{lemma}\label{L:London241228}
Suppose $X$ and $X'$ are $\R$-valued semimartingales such that $\Delta X'>-1$ and $X-X'$ is non-decreasing. Then $\frac{\Exp(X)}{\Exp(X')}$ is non-decreasing. 
\end{lemma}
\begin{proof}
By the generalised Yor formula \cite[Proposition~4.1]{cerny.ruf.22.ejp} and composition (Proposition~\ref{P:composition}), we have 
\[
\frac{\Exp(X)}{\Exp(X')}=\Exp\bigg(\bigg(\frac{1+\id_1}{1+\id_2}-1\bigg)\circ(X,X')\bigg)=\Exp\bigg(\frac{\id_1}{1+\id_2}\circ(X-X',X')\bigg).
\] 
Since $X-X'$ is of finite variation, the \'Emery formula \eqref{eq:circ tI} yields  
\[
\frac{\id_1}{1+\id_2}\circ(X-X',X') = (X-X')^c+[X',X']^c+\sum_{t\leq\cdot}\frac{\Delta (X-X')_t}{1+\Delta X'_t}
\] 
and the claim follows since the right-hand side is non-decreasing.
\end{proof}

\begin{lemma}\label{L:Cierne250106}
Let $X$ be an $\R$-valued semimartingale with independent increments starting at $0$ such that $(\id^2\indicator{\id<-1})\circ X$ is special. Then $\sup_{t\in[0,T]}X_t^-\in L^2$. 
\end{lemma}
\begin{proof}
Observe that $Y=(\id\indicator{\id\leq 1})\circ X$ has independent increments and is special. Since $\sup_{t\in[0,T]}X_t^-\leq \sup_{t\in[0,T]}Y_t^-$, it suffices to show the statement with $Y$ in place of $X$. We have $Y^-\leq (Y-B^Y)^-+(B^Y)^-$, hence
\[ \sup_{t\in[0,T]}Y_t^- \leq \sup_{t\in[0,T]}|Y_t-B^Y_t| + \sup_{t\in[0,T]}|B^Y_t|,\]
where the second term on the right-hand side is a deterministic constant. Since $\id^2\indicator{|\id|\leq1}\circ X$ has bounded jumps, we have $\id^2\circ Y = (\id^2\indicator{\id<-1})\circ X + \id^2\indicator{|\id|\leq1}\circ X$ is special.  
The Doob inequality and the fact that $\id^2\circ Y$ has independent increments now yield
\[\E\left[\left(\textstyle\sup_{t\in[0,T]}|Y_t-B^Y_t|\right)^2\right] \leq 4\E[\id^2\circ(Y-B^Y)_T] \leq 4B^{\id^2\circ Y}_T, \]
which completes the proof.
\end{proof}

The next theorem yields a simple proof that a semimartingale with independent increments is a local martingale if and only if it is a martingale, a claim previously established in \citet*[Theorem~5.8]{shiryaev.cherny.02}.

\begin{theorem}[Integrability of processes with independent increments] \label{T:250107}
For an $\R$-valued semimartingale $X$ with independent increments starting at $0$, the following are equivalent.
\begin{enumerate}[(i)]
\item\label{T:250107.1} $X^T$ is special.
\item\label{T:250107.2} $\sup_{t\in[0,T]} |X_t|\in L^1$.
\item\label{T:250107.3} $X_T\in L^1$. 
\end{enumerate}
Moreover, if one of these conditions holds, then $\E[X_t] = B^X_t$ for all $ t\in[0,T]$. In particular, if $X^T$ is a $\sigma$-martingale, then any of these conditions implies that $X^T$ is a martingale.
\end{theorem}

\begin{proof}
We first establish that the sum of two independent random variables is integrable if and only if each random variable is integrable. Indeed, let $U$ and $V$ be independent. Consider $v\geq0$ such that $\P[|V|\leq v]>0$. Then
\begin{align*}
\E[|U|] = \E\big[|U| \,\big|\, |V|\leq v\big]  = \E\big[|U+V-V| \,\big|\, |V|\leq v\big]\leq \E\big[|U+V| \,\big|\, |V|\leq v\big] + v<\infty,
\end{align*}
which asserts the claim.
\smallskip

\noindent \ref{T:250107.1} $\Rightarrow$ \ref{T:250107.2} Let $Y=|\id|\indicator{|\id|>1}\circ X^T$ and observe that $Y$ is an non-decreasing process with independent increments. By \cite[Proposition~II.2.29]{js.03}, $X^T$ is special if and only if $Y$ is special. Let $(\tau_n)_{n\in\N}$ be a localizing sequence for the local martingale $Y-B^Y$. By monotone convergence we have 
\[\E[Y_T]=\lim_{n\uparrow\infty}\E[Y_{\tau_n}]=\lim_{n\uparrow\infty}\E[B^Y_{\tau_n}]\leq B^Y_T,\]
hence $Y_T\in L^1$. Next, observe that $Z = \id\indicator{|\id|\leq 1}\circ X^T$ has independent increments with bounded jumps. Lemma~\ref{L:Cierne250106} yields $\sup_{t\in[0,T]}|Z_t|\in L^2\subset L^1$. The inequality 
$$\sup_{t\in[0,T]}|X_t|\leq Y_T+\sup_{t\in[0,T]}|Z_t|$$ 
now concludes.

\noindent \ref{T:250107.3} $\Rightarrow$ \ref{T:250107.1}
Since $X_T-X_t$ and $X_t$ are independent, we have that $X_t\in L^1$ for all $ t\in[0,T]$. For the deterministic process $B$ defined by $B_t=\E[X^T_t]$, $t\geq 0$,  $X-B$ has independent increments with zero mean, hence $X^T-B$ is a martingale with a version whose paths are c\`adl\`ag almost surely. Since $B$ is deterministic, this yields that $B$ is a finite-variation semimartingale and therefore $X^T=B+(X^T-B)$ is special. 
\end{proof}

\vspace{1ex}
\currentpdfbookmark{Acknowledgements}{Acknowledgements}
\noindent {\bf Acknowledgements.} We would like to thank two anonymous referees for their helpful comments. J.R. acknowledges financial support from the EPSRC Research Grant EP/Y024524/1.

\def\MR#1{\href{http://www.ams.org/mathscinet-getitem?mr=#1}{MR#1}}
\def\ARXIV#1{\href{https://arxiv.org/pdf/#1}{arXiv:#1}}
\def\DOI#1{\href{https://doi.org/#1}{#1}}

\end{document}